\documentclass[reqno,11pt,a4paper]{amsart}
\usepackage{graphicx}
\newtheorem{theorem}{Theorem}
\newtheorem{lemma}[theorem]{Lemma}
\newtheorem*{lemma*}{Lemma}
\newtheorem{proposition}[theorem]{Proposition}

\usepackage[utf8]{inputenc}
\usepackage[english]{babel}
\usepackage{latexsym}
\usepackage{amssymb}
\usepackage{amsmath}
\usepackage{xcolor}

\def\w#1{\mathop{:}\nolimits\!#1\!\mathop{:}\nolimits}

\newcommand{\norm}[1]{\Vert #1\Vert}

  \newtheorem{remark}[theorem]{Remark}

\theoremstyle{remark}

\renewcommand{\Re}{\mathrm{Re}\, }
\renewcommand{\Im}{\mathrm{Im}\,}

\newcommand\otimesal{\mathop{\hbox{\raise 1.6 ex
  \hbox{$\scriptscriptstyle\mathrm{al}$}
\kern -0.92 em \hbox{$\otimes$}}}}
\newcommand\oplusal{\mathop{\hbox{\raise 1.6 ex
  \hbox{$\scriptscriptstyle\mathrm{al}$}
\kern -0.92 em \hbox{$\oplus$}}}}
\newcommand\Gammal{\hbox{\raise 1.7 ex
\hbox{$\scriptscriptstyle\mathrm{al}$}\kern -0.50 em $\Gamma$}}

   \let\ep=\epsilon

  \let\ga=\gamma 
 \let\la=\lambda

 \let\Ga=\Gamma \let\La=\Lambda

\newcommand{\caA}{{\mathcal A}}
\newcommand{\caB}{{\mathcal B}}
\newcommand{\caC}{{\mathcal C}}
\newcommand{\caD}{{\mathcal D}}
\newcommand{\caE}{{\mathcal E}}
\newcommand{\caF}{{\mathcal F}}
\newcommand{\caG}{{\mathcal G}}
\newcommand{\caH}{{\mathcal H}}

\newcommand{\caK}{{\mathcal K}}
\newcommand{\caL}{{\mathcal L}}
\newcommand{\caM}{{\mathcal M}}

\newcommand{\caO}{{\mathcal O}}

\newcommand{\caR}{{\mathcal R}}
\newcommand{\caS}{{\mathcal S}}

\newcommand{\caV}{{\mathcal V}}
\newcommand{\caW}{{\mathcal W}}

\newcommand{\caY}{{\mathcal Y}}

\newcommand{\scrG}{{\mathscr G}}

\newcommand{\bbC}{{\mathbb C}}

\newcommand{\bbE}{{\mathbb E}}

\newcommand{\bbP}{{\mathbb P}}

\newcommand{\bbR}{{\mathbb R}}

\newcommand{\bbT}{{\mathbb T}}

\newcommand{\bbZ}{{\mathbb Z}}
\newcommand{\no}{{\noindent}}
\newcommand{\opunit}{\text{1}\kern-0.22em\text{l}}
\newcommand{\funit}{\mathbf{1}}

\newcommand{\frc}{{\mathfrak c}}

\newcommand{\frm}{{\mathfrak m}}

\newcommand{\frs}{{\mathfrak s}}

\newcommand{\frz}{{\mathfrak z}}

\newcommand{\frC}{{\mathfrak C}}

\newcommand{\non}{\nonumber}

\newcommand{\iu}{{\mathrm i}}

\newcommand{\beq}{ \begin{equation} }
\newcommand{\eeq}{ \end{equation} }
\newcommand{\bet}{ \begin{theorem} }
\newcommand{\eet}{ \end{theorem} }

\newcommand{\baq}{\begin{eqnarray}}
\newcommand{\eaq}{\end{eqnarray}}

 \newcounter{smallarabics}
\newenvironment{arabicenumerate}
{\begin{list}{{\normalfont\textrm{\arabic{smallarabics})}}}
  {\usecounter{smallarabics}\setlength{\itemindent}{0cm}
  \setlength{\leftmargin}{5ex}\setlength{\labelwidth}{4ex}
  \setlength{\topsep}{0.75\parsep}\setlength{\partopsep}{0ex}
   \setlength{\itemsep}{0ex}}}
{\end{list}}

\newcounter{smallroman}

\newcommand{\ben}{\begin{arabicenumerate}}
\newcommand{\een}{\end{arabicenumerate}}

\newcommand{\dist}{\mathrm{dist}}

\newcommand{{\banone}}{\scrG} 
\newcommand{{\gamzero}}{\ga_0} 
\newcommand{{\tengam}}{10\gamzero} 
\newcommand{{\fifteengam}}{15\gamzero} 
\newcommand{{\twentygam}}{20\gamzero} 
\newcommand{{\normba}}{\norm^{}_{\banone}}

\newcommand{\hf}{{_1\over^2}}

\newcommand{\E}{{\mathbb E}}
\newcommand{\Z}{{\mathbb Z}}
\newcommand{\R}{{\mathbb R}}
\newcommand{\C}{{\mathbb C\hspace{0.05 ex}}}

\newcommand{\T}{{\mathbb T}}

\title{Renormalization of Generalized KPZ equation  }

\author{Antti Kupiainen and Matteo Marcozzi}
\address{University of Helsinki, Department of Mathematics and Statistics,
         P.O. Box 68 , FIN-00014 University of Helsinki, Finland}
\email{antti.kupiainen@helsinki.fi}
\email{matteo.marcozzi@helsinki.fi}
\thanks{Supported by Academy of Finland}
\date{\today}

\begin{document}
\maketitle

\begin{abstract}
We use Renormalization Group  to prove local well posedness for a generalized KPZ equation introduced by H. Spohn in the context of  stochastic hydrodynamics.  The equation requires the addition of counter terms diverging with a cutoff $\epsilon$ as $\epsilon^{-1}$ and  $\log\epsilon^{-1}$. 
 \end{abstract}




\section{Introduction}
Nonlinear stochastic  PDE's driven by a space time white noise have been under intensive study in recent years  \cite{hairer, CC, GoJa, GIP, AK}. These equations are of the form
\beq
 \partial_t u=\Delta u+V(u)+\Xi 
  \label{eq: upde}
\eeq
where $u(t,x)\in\bbR^n$ is defined on $\La\subset\bbR^d$, $V(u)$ is a function of $u$ and possibly  its derivatives which can also be non-local  and $\Xi$ is white noise on $\bbR\times \La$,
formally
\beq
 \bbE\ \Xi_\alpha(t',x')\Xi_\beta(t,x)=\delta_{\alpha\beta}\delta(t'-t)\delta(x'-x).
  \label{eq: white}
\eeq
In order to be defined these equations in general require renormalization. One first regularizes the equation by e.g. replacing the noise by a mollified version $\Xi^{(\epsilon)}$ which  is smooth on scales less than $\epsilon$ and then replaces $V$ by  $V^{(\epsilon)}=V+W^{(\epsilon)}$ where $W^{(\epsilon)}$ is an $\epsilon$-dependent "counter term". One attempts to choose this so that  solutions converge as  $\epsilon\to 0$.

The
rationale of such counterterms is that  although they diverge as $\epsilon\to 0$ their effect on solutions on scales much bigger than $\epsilon$ is small. They are needed to make the equation well posed in small scales but they disturb it little in large scales.

Such a phenomenon is familiar in quantum field theory. For instance in quantum electrodynamics the "bare" mass and charge of the electron have to be made cutoff dependent so as to have cutoff independent measurements at fixed scales. The modern way to do this is to use the Renormalization Group (RG) method which constructs a one parameter family of {\it effective theories} describing how the parameters of the theory vary with scale.

Such a RG method was applied to SPDE's in \cite{AK} for the case $n=1$, $d=3$ and  $V(u)=u^3$. In that case $W^{(\epsilon)}=(a\epsilon^{-1}+b\log \epsilon)u$ and path wise solutions were constructed recovering earlier results by \cite{hairer, CC}. In the present paper we consider the equations of Stochastic Hydrodynamics recently introduced by Spohn \cite{Sp14}. They give rise to the problem \eqref{eq: upde} with $n=3$, $d=1$ and 
\beq
 V(u)=(\partial_xu,M\partial_xu)
  \label{eq: kpznonl}
\eeq
where $ (\cdot, \cdot)$ denotes the standard inner product in $ \R^3$ and $M=(M^{(1)},M^{(2)},M^{(3)})$ with $M^{(i)}$ are symmetric matrices, so that \eqref{eq: kpznonl} can be read component-wise as $V_i(u)=(\partial_xu,M^{(i)}\partial_xu) $ for $ i=1,2,3$. 
We construct path wise solutions in this case by taking
$$
W^{(\epsilon)}=a\epsilon^{-1}+b\log \epsilon.$$
The case $n=1$ is the KPZ equation and this was constructed before by Hairer \cite{hairerkpz}. In that case $b=0$. For a generic $M_{\alpha\beta\gamma}$  in  \eqref{eq: kpznonl} $b\neq0$. This counter term is third order in the nonlinearity as will be explained below. Thus in this case the simple Wick ordering of the nonlinearity does not suffice to make the equation well posed.

The content of the paper is as follows. In section 2 we define the model and state the result. The
RG formalism is set up in a heuristic fashion in Section 3. Section 4 discusses the leading perturbative solution and sets up the fixed point problem for the remainder. Section 5 states the estimates
for the perturbative noise contributions and in Section 6 the functional spaces for RG are
defined and the fixed point problem solved.  The main result is proved in Section 7. Finally in Sections 8 estimates for the covariances of the various noise contributions are proved.

\section{The regularized equation and main result}

We consider the equation \eqref{eq: upde} with $u(t,x)$ defined on  $(t,x)\in \bbR\times \bbT$ and nonlinearity given by \eqref{eq: kpznonl}. We study its integral form
\beq
u=G\ast[(V(u)+\Xi )\funit_{t\geq 0}] +e^{t\Delta}u_0  \label{eq:inteq0 }
\eeq
where $G(t,x)=e^{t\Delta}(x,0)$ and $u_0$ is the initial condition.  In this paper we consider a random initial condition of Brownian type. Concretely we take $u_0$ the stationary solution to the linear problem $V=0$ which is the Gaussian random field with covariance
$$
\bbE u_0(x)u_0(y)=\sum_{n\in\bbZ\setminus \{0\}} \frac{e^{2\pi in(x-y)}}{2(2\pi n)^{2}}.
$$
$\Xi$ is taken to be the white noise with vanishing spatial average i.e.
$$
\Xi(t,x)=\sum_{n\in\bbZ\setminus \{0\}}e^{2\pi inx}\dot b_n(t)
$$
with $b_n=\bar b_{-n}$ independent complex Brownian motions. Thus \eqref{eq:inteq0 } can be written in the form
\beq
u=G\ast(V(u)\funit_{t\geq 0}+\Xi )   \label{eq:inteq }
\eeq
 Instead of mollifying the noise we regularize the convolution by considering
\beq
u=G_\epsilon\ast(V^{(\epsilon)}(u)+\Xi )  \label{eq:inteqeps }
\eeq
where 
\beq
G_\epsilon(t,x)=e^{t\Delta}(x,0)(1-\chi(\epsilon^{-2}t))\label{eq:cutoff}
\eeq
with $\chi\geq 0$ being a smooth bump, $\chi(t)=1$ for $t\in[0,1]$ and $\chi(t)=0$ for $t\in[2,\infty)$ and
\beq
V^{(\epsilon)}(u)=[(\partial u,M\partial u)+C_\epsilon]\funit_{t\geq 0}  \label{eq:Feps }
\eeq
We look for $C_\epsilon$ such that  \eqref{eq:inteqeps } has a unique solution  $u^{(\ep)}$
which converges as $\epsilon\to 0$ to a non trivial limit. Note that $G_\epsilon\ast\Xi$ is a.s. smooth. 
 
 Our main result is

\begin{theorem}\label{main result}  There exits $C_\ep$ s.t. the following holds. For almost all realizations of the white noise $\Xi$ there exists  $t(\Xi)>0$  such that the equation \eqref{eq:inteqeps }
has for all $\ep>0$ a unique smooth solution $u^{(\ep)}(t,x)$,
$t\in [0,t(\Xi)]$ and there exists $u\in \caD'( [0,t(\Xi)]\times \bbT)$ such that
$u^{(\ep)}\to\ u$ in  $\caD'( [0,t(\Xi)]\times \bbT)$. The limit $u$ is independent
of the regularization function $\chi$. 
\end{theorem}

\begin{remark}\label{rem: renorm} 
We will find that the renormalization parameter is given by 
\beq
C_\epsilon=m_1\epsilon^{-1}+m_2\log\epsilon^{-1} +m_3  \label{eq: reps}
\eeq
where the constants $m_1$ and $m_3$ depend on $\chi$ whereas the  $m_2$ is universal
i.e. independent on  $\chi$.
Furthermore, $m_2=0$ if $ M^{(\alpha)}_{\beta \gamma}$ is totally symmetric in the three indices.
\end{remark}



\section{Renormalization group}


The regularized equation \eqref{eq:inteqeps } can be viewed as dealing with spatial scales larger than $\epsilon$. The idea of the Renormalization Group (RG) is to try to increase this small scale cutoff by deriving {\it effective equations} with larger cutoffs. This will be done inductively by going from scale $\ell$ to scale  $L\ell$ with with $L$ fixed. One such step is called the RG transformation. It is useful to utilize the underlying scale invariance of the linear part of the equation and rescale at each step the small scale cutoff to unity. 
To do this define the space time scaling $s_\mu$ by
\beq
 (s_\mu f)(t,x)=\mu^{-\hf} f(\mu^{2}t,\mu x)
 \non
\eeq
and set
\beq
\varphi=s_{\epsilon}u. \label{eq: phindef}
\eeq
Note that $\varphi$ is defined on  $\bbR\times \epsilon^{-1}\bbT$.
By a simple change of variables in  \eqref{eq:inteqeps } we obtain
\beq
 \label{eq:neweq1 }
\varphi= G_1\ast (v^{(\epsilon)}(\varphi)+\xi)
\eeq
where 
\beq
v^{(\epsilon)}(\varphi):=\epsilon^{\frac{1}{2}}(\partial_x \varphi,M\partial_x \varphi)+\epsilon^{\frac{3}{2}}C_\epsilon
\label{eq: vNdef}
\eeq
and $\xi:=\epsilon^{2}s_{\epsilon}\Xi $ is equal in law with the white noise on $\bbR\times \epsilon^{-1}\bbT$ (we keep the convention that $v^{(\epsilon)}(\varphi)=0$ for $t<0$) . 

We note that in these dimensionless variables the small scale cutoff is unity and the strength of the nonlinearity is small, $\epsilon^\hf$ i.e. the model is {\it subcritical}. However, the price we pay is that we need to consider times of order $\epsilon^{-2}$ and spatial box of size $\epsilon^{-1}$.

Let us now attempt to increase the cutoff $\epsilon$. Fix $L>1$ and decompose
$$
G_1=G_{L^{2}}+(G_1-G_{L^{2}})
$$
and
$$
\varphi= \varphi_1+\varphi_2.
$$
Then  \eqref{eq:neweq1 } is equivalent to the pair of equations
\baq
\varphi_1&=&G_{L^{2}}(v^{(\epsilon)}(\varphi_1+\varphi_2)+\xi)\non\\
\varphi_2&=&(G_1-G_{L^{2}})(v^{(\epsilon)}(\varphi_1+\varphi_2)+\xi).\non
\eaq
$\varphi_1$ can be thought of living on scales $\geq L$ and $\varphi_2$  on scales $\in [1,L]$. Rescale now back to unit cutoff. Let $s:=s_{L^{-1}}$ and set 
$$ \varphi_1=s \varphi',\ \ \   \varphi_2=s \zeta.
$$ 
Then
\beq
\varphi=s(\varphi'+\zeta)
\label{solutions}
\eeq
with $\varphi'$, $\zeta$ solutions to
\baq
\varphi'&=&G_1\ast(Sv^{(\epsilon)}(\varphi'+\zeta)+\xi)\label{h'}\\
\zeta&=&\Gamma\ast (Sv^{(\epsilon)}(\varphi'+\zeta)+\xi)\label{zeta}
\eaq
where we defined the scaling operation 
$$ (Sv)(\varphi)=L^{2}s^{-1} v(s \varphi)$$ 
and denoted
\beq
 \label{eq:Gammadef }
\Gamma(t,x):=
e^{t\Delta}(x,0)(\chi(t)-\chi(L^2t)).
\eeq
Note that  $\Ga$ involves scales between $L^{-1} $ and $1$ so that the equation \eqref{zeta}  turns out to be tractable: its solution  $\zeta$ is a function $\zeta(\varphi')$ of $\varphi'$. Plugging this into the large scale equation \eqref{h'} yields
\beq
 \label{eq:neweq2 }
\varphi'= G_1\ast (\caR v^{(\epsilon)}(\varphi')+\xi)
\eeq
where the new nonlinearity $\caR v^{(\epsilon)}$ is defined by
\beq
 \label{eq:neweq3 }
\caR v^{(\epsilon)}(\varphi')=S v^{(\epsilon)}(\varphi'+\zeta(\varphi')).
\eeq
$\caR$ is the {\it Renormalization Group map}: given a function $v$ mapping a field $\varphi(t,x)$ to a field $v(\varphi)(t,x)$ we obtain a new function $\caR v$ by solving the small scale equation. 
Using \eqref{zeta} in \eqref{eq:neweq3 } we may write the latter as an equation to determine $\caR v$: 
\beq
 \label{eq:neweq4 }
\caR v(\varphi)=Sv(\varphi+\Gamma\ast(\caR v(\varphi)+\xi)).
\eeq 
We will set up the functional spaces where \eqref{eq:neweq4 } is solved in Section 6. At this point let us see on a formal level how the solution of the original  SPDE is reduced to the study of the map $\caR$. To do this it is convenient to take the cutoff $\epsilon$ as
\beq
 \label{epsdef}
\epsilon=L^{-N}
\eeq 
so that we are interested in the limit $N\to\infty$. With a slight abuse of notation, denote $ v^{(\epsilon)}$ by  $ v^{(N)}$ and define inductively
\beq
 \label{eq:neweq44 }
 v^{(N)}_{n-1}:=\caR v_n^{(N)}.
\eeq
for $n=N,N-1,\dots$. 

We call $ v^{(N)}_n$ the {\it effective potential} at scale $L^{-n}$ starting with cutoff $L^{-N}$. They are related to each other by the iteration
\beq
v^{(N)}_{n-1}(\varphi)=
Sv^{(N)}_n(\varphi+\Gamma_{n}\ast(v^{(N)}_{n-1}(\varphi)+\xi_{n-1}))
 \label{eq: vn+1new}
\eeq
where we denote explicitly the dependence of the noise on the scale:
$$\xi_n:=L^{-2n}s^{-n}\Xi.$$
$\xi_n$ equals in law the white noise in $\bbR\times L^n\bbT$. $\Gamma_{n}$ is the operator \eqref{eq:Gammadef } on  $\bbR\times L^n\bbT$.
\begin{remark}\label{scaledep} The definition of $\caR$ involves the scale $n$ i.e. the size $L^n$ of the spatial box where the heat kernel in \eqref{eq:Gammadef } is defined. We suppress this dependence in the notation unless we want to emphasize it.
\end{remark}
From \eqref{solutions} we infer that solutions to the equations $v$ and $v'=\caR v$ are related by
$$
\varphi=s(\varphi'+\Gamma\ast (v'(\varphi')+\xi)).
$$
This leads to an iterative construction of the solution as follows.  Suppose  $\varphi_n$ solves the {\it effective equation}
\beq
 \label{eq:neweqn }
\varphi_n= G_1\ast (v^{(N)}_n(\varphi_n)+\xi_n).
\eeq
Then,
the solution of the original equation \eqref{eq:neweq1 }  is given by
 \beq
\varphi=s^{-(N-n)}
f^{(N)}_n(\varphi_n) .\label{eq: finalsolution}
\eeq 
where the maps $f^{(N)}_n$ satisfy the induction
\beq
f^{(N)}_{n-1}(\varphi)=
L^{-2}Sf^{(N)}_n(\varphi+\Gamma_{n}\ast(v^{(N)}_{n-1}(\varphi)+\xi_{n-1}))  \label{eq: fn+1new}
\eeq
with the initial condition
\beq
f^{(N)}_{N}(\varphi)=\varphi \label{eq:initialf }.
\eeq
Recalling \eqref{eq: phindef} we conclude that the solution of the SPDE with cutoff $\epsilon$ is given by
 \beq
u=s^{n}
f^{(N)}_n(\varphi_n) .\label{eq: finalsolution}
\eeq 
Suppose now that (a) we can control the $v^{(N)}_n$ and $f^{(N)}_n$ for $n\geq m$, (b)  we can solve \eqref{eq:neweqn } for $n=m$ on the time interval $[0,1]$ (c) the solution $\varphi_m$ is in the domain of  $f^{(N)}_m$. Then \eqref{eq: finalsolution} yields the solution of the SPDE on the time interval $[0,L^{-2m}]$. 

What determines the smallest $m$ so that (a)-(c) hold? This is determined by the {\it realization of the noise $\Xi$}. Indeed, the  $v^{(N)}_n$ are {\it random} objects i.e. functions of the  white noise $\Xi$. Let $\caE_m$ be the event such that the above holds for all $N,n$ with $m\leq n\leq N$. We will show that almost surely $\caE_m$ holds for some $m<\infty$. For a precise statement see Section 5.


Equations \eqref{eq:neweq1 }, \eqref{eq: vn+1new} and \eqref{eq: fn+1new} involve the convolution operators $\Gamma_n$ and  $G_1$ respectively. These operators are  infinitely smoothing and their kernels have fast decay in space time.
  In particular  the noise $\zeta=\Gamma_n \ast \xi_{n-1}$ entering equations  \eqref{eq: vn+1new} and \eqref{eq: fn+1new} has a smooth covariance which has
finite range in time
and it has Gaussian decay in space. Hence the fixed point problem  \eqref{eq: vn+1new} turns out to be quite easy.


\section{Perturbative contributions}

The RG iteration we have defined is quite general: formally it holds for ``arbitrary'' nonlinearity $v$ (and in any dimension as well, with appropriate scaling $s$). In the case at hand $v$ is a function of $\partial_x \varphi$ so it pays to change variables and denote 
$$
\phi:=\partial_x \varphi.
$$
Denote also 
$$
v^{(N)}_n(\varphi)=w^{(N)}_n(\phi)
$$
and redefine the scaling operation as
$$
(\frs\phi)(t,x)=L^{-\hf}\phi(L^{-2}t,L^{-1}x)
$$
and 
$$
(\caS  v)(\phi)=L \frs^{-1}v(\frs\phi)
$$
so that the RG iteration \eqref{eq: vn+1new} becomes
\beq
w^{(N)}_{n-1}(\phi)=
\caS w^{(N)}_n(
\phi+\Upsilon_n\ast(w^{(N)}_{n-1}(\phi)+\xi_{n-1}))
 \label{eq: wn+1new}
\eeq
where 
$$\Upsilon_n=\partial_x\Gamma_n
.$$
Eq.  \eqref{eq: fn+1new} in turn becomes
\beq
f^{(N)}_{n-1}(\phi)=L^{-1}
\caS f^{(N)}_n(
\phi+\Upsilon_{n}\ast(w^{(N)}_{n-1}(\phi)+\xi_{n-1}))  \label{eq: fn+1new2}
\eeq
and we have the initial conditions
\baq
 w^{(N)}_N(\phi)&=&L^{-\frac{1}{2}N} (\phi,M\phi)-L^{-\frac{3}{2}N}C_{L^{-N}}  
 \label{eq: first v}\\
  f^{(N)}_N(\phi)&=&\phi .\label{eq: first f}
\eaq
From now on to avoid too many indices we suppress in the notation the superscript $(N)$ so that $N$ is considered fixed and the scale $n$ runs down from $n=N$.


%

\subsection{Solving the first order}

It is instructive and useful 
to study the fixed point equation \eqref{eq: wn+1new} to first order in $w$. 
Define the map
\begin{align*}
 (\mathcal{L} w )(\phi):=\caS w(\phi+  \Upsilon_n\ast\xi_{n-1}).
 \end{align*}
Then  \eqref{eq: wn+1new} can be written as
\beq
w_{n-1}(\phi)=(\mathcal{L}w_n)(\phi+\Upsilon_{n}\ast w_{n-1}(\phi)) 
 \label{eq: vn+1new1}
\eeq
so $\caL$ is the {\it linearization} of the RG map  $\caR$:  $\caL= D\caR$.
 Its properties are crucial for understanding the flow of effective  equations $w_n$.
 
 Let us 
 consider the linear RG flow from scale $N$ to scale $n$ i.e. 
 $\caL^{N-n} w_{N}$. 
 This can be computed by doing one RG step with  $L$ replaced by $L^{N-n}$.  We get
 \begin{align}\label{eq:iter_linear_map}
 \caL^{N-n} w_{N}(\phi)= \caS^{N-n}w_N(\phi + Y^{(N)}_{n}\ast \xi_n)
\end{align}
where
\beq
Y^{(N)}_{n}(t,x)=\partial_xH_n(t,x)\chi_{N-n}(t).\label{eq: cayn}
\eeq
with 
\begin{align}\label{eq:heat_torus}
H_n(t,x)=\frac{1}{\sqrt{4\pi t}} \sum_{i \in \Z}e^{-\frac{(x+i L^{n})^2}{4t}}
\end{align}
being the heat kernel on $\bbT_n$ and
\begin{align}\label{eq:chiN}
\chi_{m}(s):=\chi(s)-\chi(L^{2m}s) 
\end{align}
a smooth indicator of the interval $ [L^{-2m},2]$. The field
$$\vartheta_n
:=Y^{(N)}_{n}\ast \xi_n
$$
 is a stationary Gaussian vector-valued field with covariance given by 
\beq
 \E \vartheta_{n, \alpha}(t,x) \vartheta_{n, \beta}(s,y)=\delta_{\alpha \beta}{\mathfrak{C}}^{(N)}_n(|t-s|, x-y) \label{eq: fracCn-ndef}
\eeq
where 
\beq
 {\mathfrak{C}}^{(N)}_n(t,x)=-\Delta \int_{0}^{\infty} H_n(t+2\tau,x)\chi_{N-n}(t+\tau)\chi_{N-n}(\tau)d\tau. \label{eq: fracCn-ndef1}
\eeq 

The scaling operator  has eigenfunctions
\begin{align}
\caS \phi^k&= 
L^{\frac{3-k}{2}} \phi^k.
\end{align}
From this one obtains
\beq \label{eq:linear_sol}
 \caL^{N-n} w_{N}(\phi)=
L^{-\frac{1}{2}n} (\phi+\vartheta_n,M(\phi+\vartheta_n))-L^{-\frac{3}{2}n}C_{L^{-N}}.
\eeq
We see now why the counter term $C_{L^{-N}}$ is needed: the expectation of the random field $(\vartheta_n,M\vartheta_n)$ blows up as $N\to\infty$ as shown in Lemma \ref{lemma: heatkernel} and this divergence is the source of the renormalization constant $ m_1$ in \eqref{eq: reps}.

Furthermore, we need to study the dependence of our constructions on the choise of the cutoff function $\chi$ in  \eqref{eq:cutoff}. To this end, let us define
\begin{align}
\chi'_{m}(s)= \chi(s)-\chi'(L^{2m}s)
\end{align}
where the lower cutoff in \eqref{eq:chiN} has been replaced by a different bump function $ \chi'$. 
In the following we will denote by $ Y_n'^{(N)}$ the kernel $ Y_n^{(N)}$ where $ \chi_{N-n}$ is replaced by $ \chi'_{N-n}$. We also note that, by taking $ \chi'(s)=\chi(L^2 s)$, one gets $ Y_n'^{(N)}=Y_n^{(N+1)}$, so by varying $ \chi'$ we can also study the dependence and convergence as $ N \to \infty$.   

We are now ready to state the Lemma which controls the dependence of the covariance $ \frC_n^{(N)}$ on $ N$ and $ \chi$. See the Appendix for the proof.

\begin{lemma}\label{lemma: heatkernel} Define $ m_1 \in \R^3$ by
\beq
m_1^{(\alpha)}:= \bigg(\sum_{\beta=1}^3 M^{(\alpha)}_{\beta \beta} \bigg)\frac{1}{2^{7/2}\sqrt{\pi}}\int_0^\infty s^{-3/2}(1-\chi(s)^2)ds \label{eq:m_def}
\eeq
for $ \alpha =1,2,3$.
Then 
\beq
\bbE (\vartheta^{(N)}_n,M \vartheta^{(N)}_n)=L^{N-n}m_1  +\delta_n^{(N)}
\non
\eeq
where $\|\delta_n^{(N)} \|$ is uniformly bounded in $ N$ and $n $ where $\|\cdot \|$ is the Euclidean norm in $\R^3$.  
Moreover, let $ \delta_n'^{(N)}$ be the analog of $ \delta_n^{(N)}$, where the lower cutoff function is replaced by $\chi'$. Then 
\beq
\|\delta_n^{(N)}-\delta_n'^{(N)} \|\leq Ce^{-cL^{2N}} 
\|\chi-\chi'\|_\infty.
\label{eq: deltaNnN'bound}
\eeq

\end{lemma} 

The counter term $ C_{L^{-N}} $ is then given in this linear approximation as
\beq 
C_{L^{-N}}= L^{N}m_1  \label{eq: C_linear}
\eeq 
and we end up with
\beq \label{eq:linear_sol2}
 \caL^{N-n} w_{N}(\phi)= u_{n,1}(\phi) .
\eeq
where
\beq \label{eq:linear_sol3}
u_{n,1}(\phi)= L^{-\frac{n}{2}}(
 (\phi+\vartheta_n,M(\phi+\vartheta_n))-L^{N-n}m_1) .
\eeq


\subsection{Higher order terms}

The heuristic idea of our proof is now the following. We look for the RG flow in the form
\begin{align}\label{eq:linear}
w_{n}=\sum_{i=1}^{k-1}u_{n,i}(\phi) +r_{n}
\end{align}
where $u_{n,i}$ are explicit perturbative contributions and in a suitable norm 
\begin{align}
\| u_{n,i}  \|  = \caO(L^{-\frac{i}{2}n}), \ \ \
\|r_{n}\|  =\caO(L^{-\frac{k}{2}n})
\end{align}
and we expect 
\beq 
r_{n-1}=\caL r_{n}+\caO(L^{-\frac{k+1}{2}n}).\label{reminder}
\eeq
Moreover, from our analysis of $\caL$ we also expect that
$$
\|\caL r_{n}\| \leq CL^{\frac{3}{2}}\|\caL r_{n}\|\leq  CL^{\frac{3}{2}}L^{-\frac{k}{2}n}=CL^{\frac{3}{2}-\frac{k}{2}}L^{-\frac{k}{2}(n-1)}
$$
so that \eqref{reminder} should iterate provided we take $k=4$. Hence, we should find the perturbative contributions to $w_{n}$ {\it up to order 3}. 
\begin{remark} The same heuristic idea works in general for subcritical problems. The dimensionless strength of the nonlinearity is $L^{-N\alpha}$ for some $\alpha>0$ and the norm of $\caL$ is $L^\beta$ for some $\beta>0$. Then one needs to do perturbation theory up to order $k-1$ with $k\alpha>\beta$.
\end{remark}
The $u_{n,i}$ may be computed by doing one RG step with scaling factor $L^{N-n}$
\beq
w_{n}(\phi)=L^{-\frac{n}{2}}(\phi+\vartheta_n+Y_{n}\ast w_{n}(\phi),M(\phi+\vartheta_n+Y_{n}\ast w_{n}(\phi)) )
 \label{eq: vn+1new1}
\eeq
where we dropped the superscript $ N$ also in $ Y_n^{(N)}$. 
We obtain
\begin{align*}
 u_{n,2}(\phi)& =2 L^{-n}(\phi+ \vartheta_{n},M(Y_n *u_{n,1}(\phi))
\end{align*}
and 
\begin{align*}
 u_{n,3}(\phi) =& L^{-\frac{3}{2}n}( (Y_n *u_{n,1}(\phi),M(Y_n *u_{n,1}(\phi)) \\\nonumber
 &+ 2 (\phi+ \vartheta_{n},M(Y_n *u_{n,2}(\phi))-m_2 \log L^N-m_3)
\end{align*}
where $ m_2$ and $m_3$ are constants to be determined.
To write the recursion \eqref{reminder} let us denote $w_n$ by $w$ and $w_{n-1}$ by $w'$ and similarly for the other functions. Then
\beq\label{reminder1}
r'(\phi)=\caL r(\phi+\Upsilon *w'(\phi))+\caF(r')(\phi)
\eeq
with 
\begin{align}
\caF(r')(\phi)=& u'_{1}(\phi+\Upsilon *w')-u'_1(\phi) - Du'_{1}(\phi)(\Upsilon *(u'_{1}+u'_{2}))-\hf D^2u'_{1}(\phi)(\Upsilon *u'_{1},\Upsilon *u'_{1}))\non\\
&+\caL u_{2}(\phi+\Upsilon *w')-\caL u_{2}(\phi)-D\caL u_{2}(\phi)\Upsilon *u'_1\non\\&+\caL u_{3}(\phi+\Upsilon *w')-\caL u_{3}(\phi)\\
\equiv & \caF_1(r')(\phi)+\caF_2(r')(\phi)+\caF_3(r')(\phi)\label{reminder2}
\end{align}
where $D$  is the (Frechet) derivative and on the LHS $w',u'$ are evaluated at $\phi$. \begin{remark} Note that   $u_i$ are  polynomials in $\phi$ so there is no problem in defining the derivative. In Section \ref{se:Banach_setting} we'll see that  $w$ is actually analytic.
\end{remark}

\section{Random fields} The perturbative terms $u_i$ are polynomials in $\phi$ with random coefficients. For the heuristic idea of the proof presented above to work these coefficients should not be too large. For $u_{n,1}$ these random coefficients are the random fields $\vartheta_n(t,x)$ and
\beq
u_{n,1}(0)=L^{-\frac{n}{2}}((\vartheta_n(t,x), M \vartheta_n(t,x))-L^{N-n}m_1 )
 \label{eq: noisefields0}
\eeq
 
In case  of $u_{n,2}$ and $u_{n,3}$ we don't need to consider all the coefficients. Indeed,
the discussion of previous section was based on a bound $L^{\frac{3}{2}}$ for the linearized RG operator. This is indeed its eigenvalue on constants. The next eigenvalue is $L$ on linear functions, the one after $L^\hf$ etc. Thus for $u_{n,2}$ we should be worried only about the constant and linear terms in $\phi$ and for $u_{n,3}$ only about constants. All the other terms should be irrelevant i.e. they should contract under the RG. We will now isolate these relevant terms. Let us expand
 \beq
u_{n,2}(\phi)=u_{n,2}(0)+Du_{n,2}(0)\phi+U_{n,2}(\phi) .
\label{eq:U_2}
\eeq
We get 
\beq
u_{n,2}(0)=2L^{-n}(\vartheta_n,M Y_n*u_{n,1}(0))
 \label{eq:U_n(0) }
\eeq
and 
\beq
Du_{n,2}(0)\phi=L^{-n}(\phi,MY_n*u_{n,1}(0))+L^{-n}\int_{-\infty}^t ds \int_{\T_n} dy \, \sigma_n(t,x,s,y)\phi(s,y)
 \label{eq:DU_n(0) }
\eeq
where 
\beq
(\sigma_n(t,x,s,y))_{\alpha \beta}=
4Y_n(t-s,x-y)\sum_{\gamma, \delta, \lambda} 
\vartheta_\gamma(t,x) M^{(\alpha)}_{\gamma \delta} \vartheta_{\lambda}(s,y))M^{(\delta)}_{\lambda \beta}\,.
 \label{eq: sigma_n}
\eeq
For the third order term we get
\beq
u_{n,3}(\phi)=u_{n,3}(0)+U_{n,3}(\phi) .
\label{eq:U_3}
\eeq
with
\begin{align}
 u_{n,3}(0) =&  L^{-\frac{3}{2} n}[(Y_n *u_{n,1}(0),M(Y_n *u_{n,1}(0)) \\\nonumber
 & + (\vartheta_n ,M(Y_n *u_{n,2}(0))-m_2 \log L^N - m_3].
\end{align}
Consider now the random fields $u_{n,i}(0), Du_{n,2}(0)$ with the scaling factor divided out, i.e. 
\begin{align}\label{eq: noisefields1}
\frz_{n,i}:= L^{\frac{i}{2}n}u_{n,i}(0), \ \ \ D \frz_{n,2}:= L^n Du_{n,2}(0). 
\end{align} 
Then $ \vartheta_n, \frz_{n,i}, D \frz_{n,2}$ belong to the Wiener chaos of white noise of bounded order ($\leq 4$) and their size and regularity are controlled by studying their covariances, as shown in the Section \ref{proof_prop}. 
For finite cutoff parameter $N$ these noise fields are a.s. smooth but in the limit $N\to\infty$ they become distribution valued. We estimate their size in suitable (negative index) Sobolev type norms which we now define.

The operator $(-\partial_t^2+1)^{-1}$ acts  on $L^2(\bbR)$ by convolution with the function\beq
 K_1(t)=\hf e^{-|t|}\,.
 \label{eq: K_1def}
\eeq
 and the operator $(-\Delta+1)^{-1}$ on $L^2(\bbT_n)$ is convolution with the periodization of \eqref{eq: K_1def} 
$$ K_2(x)= \sum_{i\in\Z}K_1(x+iL^n).
$$ 
Let
$$
K(t,x)=K_1(t)K_2(x).
$$
Note that convolution with $K$ is a positive operator in $L^2(\bbR\times\bbT_n)$.
We define $\caV_n$ to be be the completion of $C_0^\infty(\bbR_+\times\bbT_n)$ with the norm
\beq
\|v\|_{\caV_n}=\sup_{i}
\|K\ast v\|_
{L^2(\mathfrak{c}_i)}
\eeq
where $\mathfrak{c}_i$ is the unit cube centered at $i\in \bbZ\times(\bbZ\cap\bbT_n)$.
To deal with the bi-local field as $\sigma_n$ in \eqref{eq: sigma_n} we define for $\sigma(t,x,s,y)$ in $C_0^\infty(\bbR_+\times\bbT_n\times \bbR_+\times\bbT_n)$ 
\beq
\|\sigma\|_{\caV_n}=
\sup_{i}\sum_j\|K\otimes K\ast \sigma\|_{L^2(\mathfrak{c}_i\times \mathfrak{c}_j)}
\eeq

Now we can specify the admissible set of noise. Let $\ga>0$  and  define the sets of events $\caE_m$, $m>0$  in the probability space of the space time white noise $\Xi$
as follows. Let $\zeta^{(N)}_n$ denote any fields  $\vartheta_n, \frz_{n,i}, D\frz_{n,2}$.  The first condition on $\caE_m$ is that for all $N\geq n\geq m$
the following hold: 
\beq
\|h_{n}\zeta^{(N)}_n\|_{\caV_n}\leq L^{\ga n}\,.
\label{eq: amevent1}
\eeq
where $h_{n}$ is a smooth indicator of the time interval $[0,\tau_{n}]$, $\tau_{n}=L^{2(n-m)}$ which is introduced to localize in time the flow equation, as we will see in Section \ref{se:Banach_setting}. More precisely, $h$ is a smooth bump on $\bbR$ with $h(t)=1$ for $t\leq -L^{-2}$ and $h(t)=0$ for $t\geq-\hf L^{-2}$ and set 
$h_{k}(t)=h(t-\tau_k)$ so that  $h_{k}(t)=1$ for $t\leq\tau_{k}-L^{-2}$ and $h_{k}(t)=0$ for $t\geq\tau_{k}-\hf L^{-2}$. 

We need also to control the
$N$ and $\chi$ dependence of the  noise fields $\zeta^{(N)}_n$. We can study both by varying the lower cutoff in the operator $Y^{(N)}_{n}$ in  \eqref{eq: cayn}. We denote by  $\zeta'^{(N)}_n$ any of the resulting noise fields.  Our second condition on  $\caE_m$ is that for all $N\geq n\geq m$  and all cutoff functions $\chi, \chi'$ with bounded $C^1$ norm
\beq
\|h_{n}(\zeta'^{(N)}_n-\zeta^{(N)}_n)\|_{\caV_n}\ \leq  L^{-\ga(N-n)} L^{\ga n}. 
\label{eq: amevent2}
\eeq 
The final condition concerns the fields $\Upsilon_n * \xi_{n-1}$ entering the RG iteration \eqref{eq: wn+1new}. Note that these fields are $N$ independent and smooth and we are going to impose on them a smoothness condition:  for all $n> m$ we demand
\beq
 \| \Upsilon_n * \xi_{n-1}\|_{\Phi_{n-1}}\leq L^{\ga n}.
 \label{eq: Gammaxibound}
\eeq 
where the norm is defined  in next section. In Section \ref{proof_prop} we prove

\begin{proposition}\label{prop: mainproba} There exists $\ga>0$ such that almost surely $\caE_m$ holds for some $m<\infty$.
\end{proposition} 
In the following sections we suppose the noise is on  $\caE_m$ and  we will  control the RG iteration \eqref{reminder1} for scales $n\geq m$.

\section{Banach space setup for the RG map}\label{se:Banach_setting}

In this section we set up the RG iteration in suitable functional spaces along the same lines of \cite{bk, bgk}. Let us first discuss the domain and range of the effective nonlinearities  $w_n$. The range of  $w_n, r_n$ is dictated by the noise, so we take it to be  $\caV_n$.

In the argument  of $w^{}_n$  in \eqref{eq: wn+1new}   $\Upsilon_{n }*(w_{n-1}+\xi_{n-1})$ is smooth so we take the domain  of  $w_n(\phi)$  to consist of suitably smooth functions.   
Let $\Phi_n$  be the space of 
$$\phi:[0,\tau_{n}]\times\bbT_n\to\bbC$$
 which are $C^2$
 in
$t$ and $C^2$ in $x$ with $ \partial^i_t\phi(0,x)=0$ for $0\leq i\leq 2$ 
and all $x\in \bbT_n$.
We equip $ \Phi_n$
 with
the sup norm 
$$\|\phi\|_{\Phi_n}:=\sum_{i\leq 2, j\leq 2}\|\partial_t^i\partial_x^j\phi\|_\infty.
$$ 
The following lemma
collects some elementary facts on how our spaces tie up with the operators entering the RG:

\begin{lemma}\label{lem: gammamap}  
(a) $ \Upsilon_n 
: \caV_{n-1}\to \Phi_{n-1}$ and $h_{n-1}\Upsilon_n:\caV_{n-1}\to\caV_{n-1}$
are bounded operator with norm $C(L)$.
\vskip 2mm
\noindent(b) $\frs: \Phi_{n-1}\to \Phi_n$ and $\frs^{-1}: \caV_{n}\to \caV_{n-1}$ are bounded with
$$ 
\|\frs\|\leq L^{-\hf},\ \ \ \|\frs^{-1}\|\leq CL^{\hf}.$$ 
\vskip 2mm
\noindent(c) Let $\phi \in C^{2,2}(\R \times \T_n)$ and  $v\in\caV_n$. Then $\phi v\in\caV_n$ and $\|\phi v\|_{\caV_n}\leq C\|\phi\|_{C^{2,2}}\|v\|_{\caV_n}$. 

\end{lemma}
\noindent{\it Proof}. Essentially the same as Lemma 9 in \cite{AK}.
\qed

\vskip 2mm

Consider now our fixed point problem
\beq
w_{n-1}(\phi)=\caS w_n(\phi+\Upsilon_{n} * \xi_{n-1}+\Upsilon_{n}*w_{n-1}(\phi)).
 \label{eq:wn-1}
\eeq
$w_{n}$ takes values in the distribution space $\caV_{n}\subset \caD'(\R_+\times\T_n)$. We want to bound it on the time interval $[0,\tau_{n}]$ i.e. we need to localize \eqref{eq:wn-1} in time. Define 
$$\tilde w_n=h_{n}w_n$$
 so that  
\beq
\tilde w_{n-1}(\phi)=h_{n-1}\caS w_n(\phi+\Upsilon_{n} * \xi_{n-1}+\Upsilon_{n} * w_{n-1}(\phi))
 \label{eq: vn+1new1tildeee}
\eeq
One can readily check that  $\Upsilon_n w_{n-1}= \Upsilon_n \tilde w_{n-1}$ on the time interval $[0,\tau_{n-1}]$ and that 
$$h_{n-1}\caS w_n=h_{n-1}\caS \tilde w_n.$$
Thus \eqref{eq: vn+1new1tildeee} can be written as
\beq
\tilde w_{n-1}(\phi)=h_{n-1}\caS \tilde w_n(\phi+\Upsilon_{n} * \xi_{n-1}+\Upsilon_{n}*\tilde w_{n-1}(\phi)).
 \label{eq: vn+1new1tildeee1}
\eeq
We will solve \eqref{eq: vn+1new1tildeee1} in a space of analytic functions which we discuss next. Let $\caH,\caH'$ be Banach spaces and $B(r)\subset \caH$ open ball of radius $r$. Let $H^\infty(B(r),\caH')$ denotes the space of analytic functions $f:B(r)\to \caH'$ with sup norm which we denote by  $|||\cdot|||_{B(r)}$. We will use the following simple facts that are  identical 
to those of analytic functions on finite dimensional
spaces (see \cite{cha}). 

\vskip 0.2cm

\no{\bf (a)}. Let $w\in  H^\infty(B(r),\caH')$ and  $w'\in  H^\infty(B(r'),\caH'')$. If $|||w
|||< r' $
then $w'\circ w\in H^\infty(B(r),\caH'')$ and
\beq
|||w'\circ w|||_{B(r)}\,\leq\, |||w'|||_{B(r')}.
\label{circ}
\eeq
\vskip 0.2cm

\no{\bf (b)}. Let $w\in  H^\infty(B(r),\caH')$ and $\rho<r$. Then
\beq
\sup_{\Vert x\Vert< \rho } \Vert Dw(x)\Vert_{_{\caL(\caH,\caH')}}
\leq (r-\rho)^{-1} |||w|||_{B(r)},
\label{circ0}
\eeq
where ${\caL(\caH,\caH')}$ denotes the space of bounded linear operators  
from $\caH$ to $\caH'$. Taking $\rho=\hf r'$, we infer that if $|||w_i|||_{B(r)}
\leq \hf \rho$ then
\beq
|||w'\circ w_1-w'\circ w_2|||_{B(r)}\,\leq\,{\frac{_2}{^{r'}}}
\,|||w'|||_{B(r')}\;|||w_1-w_2|||_{B(r)}.
\label{circ1}
\eeq

\no{\bf (c)}. 
Define $\,\delta_k w(x) := w(x)
-\sum\limits_{\ell=0}^{k-1} \frac{1}{\ell !}\, D^\ell w(0)(x)$. 
Then
\beq
|||\delta_k 
w|||_{B(ar)}\leq\, 
{\frac{_{a^k}}{^{1-a}}} 
 |||w|||_{B(r)}
\label{circ2}
\eeq
for $0\leq a < 1$.

\vskip 2mm
Furthermore, we infer this important corollary from Lemma \ref{lem: gammamap}:


\begin{proposition}\label{prop: linrgnorm} 
$\caS$ maps $H^\infty(B(R),\caV_n )$ into $H^\infty(B(L^{\hf}R), \caV_{n-1})$ 
with norm $\|\caS\|\leq CL^{\frac{3}{2}}$. Here $B(R)\subset \Phi_n$ and $B(L^{\hf}R)\subset \Phi_{n-1}$ respectively.
\end{proposition}

\vskip 0.2cm

Let now $\gamma>0$ and set $B_n=B(L^{2\ga n})\subset\Phi_n $. Then we have

\begin{proposition}\label{prop: solution of fp}  There exist $L_0>0$, $\ga_0>0$ so that for $L > L_0$, $\ga<\ga_0$  and $m>m(\ga,L) $
 if 
 $\ \Xi\in \mathcal{A}_m\ $
  then  then for all $N\geq n-1\geq m$ 
 the equation  \eqref{eq: vn+1new1tildeee1} has a  unique solution $\tilde w^{(N)}_{n-1}\in H^\infty(B_{n-1},\caV_{n-1})$. 
 These solutions satisfy
 \beq
 |||\tilde w^{(N)}_{n}|||_{B_n}\leq L^{-\frac{1}{4} n}\label{eq: solball}
\eeq
and $\tilde w_n^{(N)}$ converge in $H^\infty(B_{n},\caV_{n})$ to a limit $\tilde w_n$ as $N\to\infty$. Furthermore, $\tilde w_n$ is independent on the small scale cutoff.
\end{proposition}

\begin{proof} 
We will drop the tilde from now on so that  $w_n$, $r_n$ and $ u_{n,i}$ stand  for 
$\tilde w_n$ etc. Also, if no confusion arises we let $w$ and $w'$ stand for $w_n$ and $w_{n-1}$ respectively.  We start with the perturbative contributions $u_i$. 
As a corollary of Lemma \ref{lem: gammamap}(c) and \eqref{eq: amevent1} we obtain for $n\geq m$ and $N\geq n$:
\begin{align}
|||u_{n,1}|||_{RB_n}& \leq  CR^2 L^{(4\ga -\hf)n}
\label{eq:unormbound1}
\\
|||u_{n,2}(0)+Du_{n,2}(0)\phi|||_{RB_n}&\leq CR L^{(3\ga-1) n}
\label{eq:Unormbound1}\\
||u_{n,3}(0)||_{\caV_n}& \leq  C L^{(\ga-\frac{3}{2}) n}
\end{align}
for all $R\geq 1$. We used also $\|h_n\|_{C^{2,2}}\leq C$. 

We need to bound the remainder terms  $U_2$ and $U_{3}$ in \eqref{eq:U_2} and  \eqref{eq:U_3}. We do this inductively in $n$. We have
$$
u_2'(\phi)=\caL u_2(\phi)
+Du_1'(\phi)\Upsilon_{n}u_1'(\phi):=\caL u_2(\phi)+v_2(\phi)
$$
Using Lemma \ref{lem: gammamap}(a), \eqref{eq:unormbound1} and \eqref{circ0} we get
$$
|||v_2|||_{L^\hf B_{n-1}}\leq C(L)L^{(6\gamma-1) n}.
$$
Let us inductively assume
\beq
\label{eq:u2ind}
|||u_{n,2}|||_{B_{n}}\leq CL^{(7\gamma -1)n}.
\eeq
Using Proposition \ref{prop: linrgnorm} and \eqref{eq: Gammaxibound} we get the following useful result
\beq
|||\caL W|||_{L^\hf B_{n-1}}\leq CL^{\frac{3}{2}}|||W|||_{B_n}.
\label{eq: caLbound}
\eeq
for all $W\in H^\infty(B_{n},\caV_{n})$ since $B(L^{2\gamma(n-1)}+ L^{\gamma n})\subset B_n$ if $L>L(\gamma)$. Thus
$$
|||u'_2|||_{L^\hf B_{n-1}}\leq C L^{\frac{3}{2}}|||u_2|||_{B_n}+ C(L)L^{(6\gamma -1)n}\leq  C L^{\frac{3}{2}}L^{(7\gamma-1) n}
$$
if $n>n(\gamma,L)$.  Then by \eqref{circ2}
$$
|||U'_2|||_{B_{n-1}}=
|||\delta_2u_2'|||_{B_{n-1}}\leq CL^\hf L^{(7\gamma -1)n}.
$$
Using \eqref{eq:Unormbound1}, the bound \eqref{eq:u2ind} follows for $n-1$ provided we take $\gamma$ so that $\,  \hf+(7\gamma-1)<0$.

 For $u_{n,3}$ we have the recursion
 \begin{align}
u'_{3}(\phi)& =\caL u_{3}(\phi)+v_3(\phi)
\label{eq:Uiteration and ic }
\end{align}
with 
$$
v_3(\phi)=\hf D^2u_1'(\phi) (\Upsilon u_1',\Upsilon u_1')+Du_1'(\phi)\Upsilon u_2'+D\caL u_2(\phi)\Upsilon u_1'.
$$
We readily get
$$
|||v_3|||_{L^\hf B_{n-1}}\leq C(L)L^{(8\gamma-\frac{3}{2}) n}.
$$
The inductive bound
\beq
\label{eq:u3ind}
|||u_{n,3}|||_{B_{n}}\leq CL^{(9\gamma-\frac{3}{2})n}.
\eeq
follows then in the same way as for $u_2$, using $U_3=\delta_1u_3$.

Now we are ready to solve equation \eqref{reminder1} by Banach fixed point theorem. Thus consider the map
\beq\label{reminder3}
\caG(r')=\caL r(\phi+\Upsilon *w')+\caF(r')
\eeq
where $\caF(r')$ is given by   \eqref{reminder2}. 

We have
$$
\caF_1(r')(\phi)=L^{-\frac{n-1}{2}}(2(\phi+\vartheta,M(\Upsilon *(u_3'+r'))+(\Upsilon *(u_2'+u_3'+r'),M \Upsilon *(u_2'+u_3'+r')))
$$
so that
$$
|||\caF_1(r')|||_{B_{n-1}}\leq C(L)(L^{(14\gamma-2)(n-1)}+L^{(2\gamma-\hf)(n-1)}|||r'|||_{B_{n-1}}).
$$
Next we write
\begin{align}
\caF_2(r')&=\caL u_{2}(\phi+\Upsilon *w')-\caL u_{2}(\phi+\Upsilon *u_1')\non\\&+\caL u_{2}(\phi+\Upsilon *u_1')-\caL u_{2}(\phi)-D\caL u_{2}(\phi)\Upsilon *u'_1
\non\\&\equiv \caF_{2,1}(r')+\caF_{2,2}(r')
\end{align}
Using   \eqref{circ1} we obtain
$$
|||\caF_{2,1}(r')|||_{B_{n-1}}\leq C(L)(L^{(14\gamma-2)(n-1)}+L^{(7\gamma-1)(n-1)}|||r'|||_{B_{n-1}})
$$
To bound $\caF_{2,2}(r')$ consider the function $f(z)= \caL u_{2}(\phi+z\Upsilon *u_1')$ for $z\in\C$. Since $ \caL u_{2}$ is analytic in $L^\hf B_{n-1}$ and 
$$\|\phi+z\Upsilon *u_1'\|_{\Phi_{n-1}}\leq L^{2\gamma(n-1)}+C L^{(4\gamma-\hf)(n-1)}|z|$$
we get that $f$ is analytic in the ball $|z|\leq C  L^{(\hf-2\gamma)(n-1)}$. Since $\caF_{2,2}(r')(\phi)=f(1)-f(0)-f'(0)$ we conclude by a Cauchy estimate
$$
|||\caF_{2,2}(r')|||_{B_{n-1}}\leq C(L)L^{(15\gamma-2)(n-1)}.$$
For $\caF_{3}$ we get using \eqref{eq: caLbound}  and \eqref{circ1}
$$
|||\caF_{3}(r')|||_{B_{n-1}}\leq C(L)(L^{(13\gamma-2)(n-1)}+L^{(9\gamma-2)(n-1)}|||r'|||_{B_{n-1}})
$$
Consider finally the first term in \eqref{reminder3}. \eqref{eq: caLbound} implies 
$$
|||\caL r(\cdot+\Upsilon *w')|||_{B_{n-1}}\leq CL^{\frac{3}{2}}|||r|||_{B_{n}}.
$$
We conclude that by taking $\gamma$ small enough if 
\beq
|||r|||_{B_{n}}\leq L^{-\frac{7}{4}n}\label{rbnbound}
\eeq
then $\caG$ maps the ball 
$
|||r'|||_{B_{n-1}}\leq L^{-\frac{7}{4}(n-1)}
$
to itself. It is now straightforward to check that $\caG$ is a contraction in this ball so that by induction in $n$ \eqref{rbnbound} holds for all $n\geq m$.

Let us address the convergence as $N\to\infty$ and cutoff dependence of $w_n=w_n^{(N)}$ which can be dealt with together by considering the difference $w_n -w'_n$ where $w'_n$ equals $w^{N+1}_n$ or $w^{N}_n$ with a different cutoff. We proceed as with $w_n$, starting with the following bounds that follow from  \eqref{eq: amevent2}:
 for all
$n\geq m$ and $N\geq n$
\begin{align}
|||u_{n,1}-u'_{n,1}|||_{RB_n}& \leq  CR^2 L^{-\ga (N-n)} L^{(4\ga-\hf)n}
\label{eq:unormbound2}
\\
|||u_{n,2}(0)+Du_{n,2}(0)\phi-(u'_{n,2}(0)+Du'_{n,2}(0)\phi)|||_{RB_n}&\leq CR L^{-\ga (N-n)}L^{(3\ga-1)n}
\label{eq:Unormbound2}\\
||u_{n,3}(0)-u'_{n,3}(0)||_{\caV_n}& \leq  C L^{-\ga (N-n)}L^{(\ga-\frac{3}{2})n}
\end{align}
for all $R\geq 1$. The induction then goes as for $w^{N}_n$, except for the prefactor $ L^{-\ga (N-n)}$ in all the bounds. This establishes the convergence of $ w^{(N)}_{n}$ to a limit that is independent on the 
short time cutoff.
\end{proof}

\section{Proof of Theorem \ref{main result} } \label{se:main result}

We can now construct the solution $\varphi^{(\ep)} $ of the $\ep$ cutoff equation
\eqref{eq:neweq1 } and consequently $ u^{(\ep)}$ in \eqref{eq:inteqeps }.  Recall that  formally $u^{(\ep)} $ is given on time interval
$[0,L^{-2m}]$ by equation \eqref{eq: finalsolution} with $n=m$ and $\varphi_m$ is the solution of equation 
\eqref{eq:neweqn } on  time interval
$[0,1]$. Hence we first need to study the $f$ iteration equation  \eqref{eq: fn+1new} which is equivalent to \eqref{eq: fn+1new2}.
We study instead of  \eqref{eq: fn+1new2} the localized iteration
\beq
\tilde f^{(N)}_{n-1}(\phi)=h_{n-1} 
L^{-1}\caS \tilde f^{(N)}_n(\phi+\Upsilon_{n}*(\tilde w^{(N)}_{n-1}(\phi)+\xi_{n-1}))
 \label{eq: fn+1new1}
\eeq
for
$\tilde f_{n}^{(N)}=h_{n} f_{n}^{(N)}$. Then we can show the following Proposition.

\begin{proposition}\label{prop: solution of fiteration} 
Let $\tilde w_{n}^{(N)}\in H^{\infty}(B_n, \caV_n)$, $m\leq n\leq N$ be as in 
Proposition \ref{prop: solution of fp} and $ \phi \in B_n$. Then for $m\leq n\leq N$ $\tilde f_{n}^{(N)}\in H^\infty(B_{n},\caV_{n})$ and  $\tilde f_{n}^{(N)}(\phi)$ converges in $ H^\infty(B_{n},\caV_{n})$   as $N\to\infty$ to a limit $ \psi_n$ which is independent of the  cutoff function.
\end{proposition} 

\begin{proof} 
Let us write
$$
\tilde f_{n}^{(N)}(\phi)=h_n(\phi+\vartheta_{n}^{(N)})+ g_{n}^{(N)}(\phi).
$$
Then
$$
 g_{n-1}^{(N)}(\phi)=h_{n-1}L^{-1}\caS  g^{(N)}_n(\phi+\Upsilon_{n}*(\tilde w^{(N)}_{n-1}(\phi)+\xi_{n-1}))+h_{n-1}\Upsilon_{n} * \tilde w^{(N)}_{n-1}(\phi).
$$
Note that the operator $L^{-1}\caS $ has norm bounded by $CL^\hf$ and $\tilde w^{(N)}_{n}$ has norm bounded by $CL^{(-\hf+4\gamma)n}$. Hence we need to extract the leading ``marginal'' part from $\tilde w^{(N)}_{n}$:
$$
g^{(N)}_{n}=h_n Y^{(N)}_{n} * \tilde u_{n,1}^{(N)}+b^{(N)}_{n}.
$$
As we will see in Section \ref{proof_prop} (Lemma \ref{lem: greg}), uniformly in $ n$ we have
$$ \caY_n :=\sup_{N \geq n} |Y_n^{(N)}| \in L^p(\R \times \T_n), \ \ \ \|Y_n'^{(N)}-Y_n^{(N)} \|_p^p \leq C L^{-\lambda(N-n)} \| \chi'-\chi \|_{\infty}
$$ for $ p < \frac{3}{2}$ and some $ \lambda >0$. Then thanks to Lemma \ref{lem: gammamap}(c) and Young's inequality we have
\begin{align}\nonumber
\|h_n Y^{(N)}_{n} * \tilde u_{n,1}^{(N)} \|_{\caV_n} \leq & C\sum_{i} \| K * Y_n^{(N)} * \tilde u_{n,1} \|_{L^2(\frc_i)} \leq C \sum_{i} \| \caY_n  \|_{L^1(\frc_i)} \| K * \tilde u_{n,1} \|_{L^2(\frc_i)} \\\nonumber
 \leq & C \| \caY_n  \|_{L^1(\R \times \T_n)} \sum_{i}  \| K * \tilde u_{n,1} \|_{L^2(\frc_i)} \\\nonumber
 \leq & C \|  \tilde u_{n,1} \|_{\caV_n}.
\end{align}
Thus $ ||| h_n Y^{(N)}_{n} * \tilde u_{n,1}^{(N)} |||_{B_n} \leq C L^{(4\gamma-\frac{1}{2})n} $ by Proposition \ref{prop: solution of fp}. 
Then the iteration of $b^{(N)}_{n}$ gives easily $|||b^{(N)}_{n} |||_{B_n}\leq L^{-\frac{3}{4}n}$, which implies that $ ||| g^{(N)}_{n}|||_{B_n} \leq L^{-\frac{1}{4}n}$.

The convergence and cutoff independence follows from that of $ \tilde w_n^{(N)}$ proved in Proposition \ref{prop: solution of fp}.
\end{proof}



Moreover, we need also this technical Lemma.
\begin{lemma}\label{lem: gammamap1}  
$ \partial_x G_1  $ is a bounded operator from $\caV_n$  to $\Phi_n$ and
 $\partial_x G_1 \ast (h_{n-1}(L^{-2}\cdot) v)=\partial_x G_1 \ast v $.
\end{lemma}
\begin{proof}
As in \cite{AK}, Lemma 14.
\end{proof}

Now we can finally prove our main result: 
let $\phi_n\in\Phi_n$ be defined inductively by $\phi_m=0$ and for $n>m$
\beq
 \label{eq:phinplus1 }
\phi_{n}= \frs(\phi_{n-1}+\Upsilon_n *(\tilde w^{(N)}_{n-1}(\phi_{n-1})+\xi_{n-1})).
\eeq
We claim that for all $m\leq n\leq N$ $\phi_n\in B_n$ and
\beq
 \label{eq:iteration for solution }
\phi_{n}= \partial_x G_1 \ast (\tilde w^{(N)}_{n}(\phi_{n})+\xi_{n}).
\eeq
Indeed, this holds trivially for $n=m$ since the RHS vanishes identically on $[0,1]$.
Suppose $\phi_{n-1}\in B_{n-1}$ satisfies
\beq
 \label{eq:neweq1mod }
\phi_{n-1}= \partial_x  G_1 \ast (\tilde w^{(N)}_{n-1}(\phi_{n-1})+\xi_{n-1}).
\eeq
Then, first by Lemma \ref{lem: gammamap}(b) and \eqref{eq:phinplus1 } 
\beq
 \non
\|\phi_{n}\|_{\Phi_{n}}\leq L^{-\hf} \|\phi_{n-1}\|_{\Phi_{n-1}}+C(L)L^{\ga n}\leq L^{2\ga n}
\eeq
so that 
 $\phi_{n}\in B_{n}$.  Second, we have
by
 \eqref{eq:neweq1mod }, \eqref{eq:phinplus1 } and Lemma \ref{lem: gammamap1}
\begin{align}\non
\phi_{n}&= \frs((\partial_x G_1+\Upsilon) \ast (\tilde w^{(N)}_{n-1}(\phi_{n-1})+\xi_{n-1})) \\
&= \partial_x G_1 \ast (\tilde w^{(N)}_{n}(\phi_{n})+\xi_{n}).\label{eq:phinplus1' }
\end{align}
Since $\phi_m=0$, from \eqref{eq: fn+1new1} we have 
\beq
\tilde f^{(N)}_{m}(0)=h_{m} 
\frs^{-1} \tilde f^{(N)}_{m+1}
(\phi_{m+1})=h_{m} h_{m+1}(L^2\cdot)\frs^{-2}\tilde f^{(N)}_{m+2}(\phi_{m+2})=h_{m} 
\frs^{-2} \tilde f^{(N)}_{m+2}(\phi_{m+2}),
\non
\eeq
then, iterating we get
\beq
\tilde f^{(N)}_{m}(0)=h_{m} 
\frs^{-(N-m)}\tilde f^{(N)}_N(\phi_N)=h_{m}h_{N}(L^{-2(N-m)}\cdot)\frs^{-(N-m)}\phi_N=h_{m}\frs^{-(N-m)}\phi_N
\label{eq:tildefn0 }
\eeq
since $ f^{(N)}_N(\phi_N)=\phi_N $ by \eqref{eq: first f}.
Now $\phi_N\in B_N$ solves \eqref{eq:iteration for solution } with
 $\tilde w^{(N)}_N(\phi)=h_{N} w^{(N)}_N(\phi)$ with $w^{(N)}_N$ given by \eqref{eq: first v}.
Since $h_{N} =1$ on $[0,\tau_{N}-L^{-2}]$ we obtain
\beq
\non
\phi_{N}= \partial_x G_1 \ast ( \tilde w^{(N)}_{N}(\phi_{N})+\xi_{N})=\partial_x G_1 \ast ( w^{(N)}_{N}(\phi_{N})+\Xi).
\eeq
To take the limit $ N \to \infty$ we will use \eqref{eq:tildefn0 }:  
defining $ \eta^{(N)} := \frs^{-N}\phi_{N}$, then we get
\beq
\non
\eta^{(N)}=\frs^{-m}\tilde f^{(N)}_{m}(0)
\eeq
on the time interval $[0,\hf L^{-2m}] \subset [0,1]$.

By Proposition \ref{prop: solution of fiteration} 
$\tilde f^{(N)}_m(
0)$ converges in $\caV_m$ to a limit $\psi_m$ which is independent of the
short distance cutoff.
Convergence in  $\caV_m$ implies convergence in $\mathcal{D}'([0,1]\times \bbT_m)$. 
The claim follows from continuity of $\frs^{-m}:\mathcal{D}'([0,1]\times \bbT_m)\to \mathcal{D}'([0,L^{-2m}]\times \bbT_1)$ and from the fact that convergence of $ \eta = \partial_x u$ implies convergence of $ u$.
 \qed 

\section{Proof of Proposition \ref{prop: mainproba}} \label{proof_prop}

We now need to show that for some $ \gamma >0$ the conditions defining the set $ \mathcal{E}_m$ hold almost surely for some $ m <\infty$. To do this, as in \cite{AK} the strategy is to control the 
covariances of the various fields in \eqref{eq: noisefields1} and establish
enough regularity for them. 

We will deduce Proposition \ref{prop: mainproba} from a covariance bound for  the fields in \eqref{eq: noisefields1}.  Let $\zeta_n^{(N)}(t,x)$ or  $\zeta_n^{(N)}(t,x,s,y)$ be any of the fields in \eqref{eq: noisefields1}. 
Let
  \beq
 \tilde K_n(t',t,x)=e^{\hf\dist(t',I_n)}K(t'-t,x)h_{n}(t)
 \label{eq:ktildedefi}
\eeq
where $I_n=[0,L^{2(n-m)}]$ and define 
 $$\rho^{(N)}_n=\tilde K_n \zeta_n^{(N)} \ \ \ \mbox{or}\ \ \ \rho^{(N)}_n=\tilde K_n\otimes \tilde K_n
 \zeta_n^{(N)}
.$$
Then
  \beq
\|K\tilde\zeta^{(N)}_n\|_{L^2(\mathfrak{c}_i)}\leq C e^{-\hf \dist (i_0,I)}\|\rho^{(N)}_n\|_{L^2(\mathfrak{c}_i)}.
\label{eq:rhozeta}
\eeq
where $i_0$ is the time component of the center of the cube $ \frc_i$.  
From now on in the random fields we will drop the superscrip $ (N)$ referring to the ultraviolet cutoff and we recall that $ \norm{\cdot}$ indicates the euclidean norm for a three-dimensional vector and the Hilbert-Schimdt norm for $ 3 \times 3$-matrix.
The following proposition proved in Section \ref{covariance_bounds} provides bounds for the covariance of $\rho_n$.  

\begin{proposition}\label{pro: proba1}
There exist renormalization constants $m_1, m_2, m_3 \in \R^3$ and $\la>0$ such that for all $0\leq n\leq N<\infty$ and for some constant $0 < c < \hf$
\baq
\bbE\norm{\rho_n(t,x)}^2&\leq & C \label{eq: twopoint}\\
\bbE \norm{{\rho'}_n(t,x)-\rho_n(t,x)}^2&\leq & C L^{-\nu (N-n)}\|\chi-\chi'\|_\infty
  \label{eq: twopoint1}\\
  \bbE \norm{\rho_n(t,x,s,y}^2&\leq & C e^{-c(|t-s|+|x-y|)}  \label{eq: twopoint2}\\
\bbE \norm{{\rho'}_n(t,x,s,y)-\rho_n(t,x,s,y)}^2 &\leq & CL^{-\la (N-n)}e^{-c(|t-s|+|x-y|)} \|\chi-\chi'\|_\infty
  \label{eq: twopoint3}
\eaq
where $ \rho'_n = \tilde K \zeta'_n$, i.e. we replace the lower cutoff function $ \chi$ by a $\chi' $.
\end{proposition} 

Now we can prove Proposition \ref{prop: mainproba}: we recall that we want to show that there exist $ 0 \leq m < \infty$ such that the event $ \caE_m$ holds almost surely, where $ \caE_m$ is the event such that bounds \eqref{eq: amevent1}, \eqref{eq: amevent2} and \eqref{eq: Gammaxibound} hold for any $ m \leq n \leq N$.  By using the same strategy as in \cite{AK} based on the bounds in \cite{boga, nualart}, one can see that Proposition \ref{pro: proba1} implies the following bounds for the random fields $ \zeta_n^{(N)}$ in \eqref{eq: noisefields1} for all $ p >1$
\begin{align}
\bbP(\|\tilde\zeta_n^{(N)}\|_{\caV_n}\geq L^{\ga n}) & \leq C L^{-2m}L^{(3-2 \gamma p)n} \label{eq:bound_local}\\
\bbP \bigg( \|\tilde{\zeta'}_n^{(N)}-\tilde\zeta_n^{(N)}\|_{\caV_n} & \geq L^{-\hf\ga(N-n)} L^{\ga n} \bigg)
\leq C L^{-p\ga(N-n)} L^{(3-2\ga p) n}L^{-2m} \label{eq:bound_bilocal}
\end{align}
Furthermore, to deal with the last condition on $\caE_m$ in \eqref{eq: Gammaxibound}, we note that $\zeta:=\Upsilon_n *\xi_{n-1}$ is a Gaussian field with covariance
$$
\bbE\zeta(t',x')\zeta(t,x)=
-\Delta_{x'}\int_0^{\infty}H_n(t'-t+2s,x'-x)
\chi(t'-t+s)\chi(s)ds\label{eq:zetacov}
$$
where $\chi$ is smooth with support in $[L^{-2},2]$.
$
\bbE\zeta(t',x')\zeta(t,x)$ is smooth, compactly supported in $t'-t$ and exponentially decaying
in $x'-x$. We get then by standard Gaussian estimates \cite{boga} for $0 \leq j \leq 2$ and $0\leq j' \leq 2$ and for some $ c(L) >0$
\begin{align}\label{eq:bound_Ups_xi}
 \bbP \bigg(\sup_{\alpha}\|\partial_t^j\partial_x^{j'} (\Upsilon_n *\xi_{n-1})_{\alpha}\|_{L^\infty(\frc_i)}>R \bigg)\leq C e^{-c(L)R^2}
\end{align}
and thus 
\begin{align}
 \bbP \bigg(\sup_{\alpha} \|(\Upsilon_n * \xi_{n-1})_{\alpha}\|_{\Phi_{n}}> L^{2\ga n} \bigg) & \leq C L^{-2m} L^{3n}e^{-c(L) L^{4\gamma n}}\,.
 \label{eq: gaxiest}
\end{align}
The bounds \eqref{eq:bound_local}, \eqref{eq:bound_bilocal} and \eqref{eq: gaxiest} implies that $ \bbP(\caE^c_m) \leq C L^{-2m}$, then 
Proposition \ref{prop: mainproba} follows from Borel-Cantelli Lemma.

\subsection{Proof of Proposition \ref{pro: proba1}}\label{covariance_bounds}

We will now study the random fields in \eqref{eq: noisefields1}, i.e. 
$$ \zeta_n\in \{ \vartheta_n, \frz_{n,i}, D\frz_{n,2}\}
$$
that enter the probabilistic estimates. 

Consider first their expectations. Setting $ z=(t,x)$ and using Lemma \ref{lemma: heatkernel}, the first one gives $ \E \frz_{n,1}=\delta_n \leq C$, while for the second order fields we have
 \begin{align}\label{bilocal2}
 \E \frz_{n,1}=\delta_n, \ \ \ \E\sigma_{\alpha \beta,n}(z,z')=\frm_{\alpha\beta} Y_n(z-z')\frC_n(z-z') 
 \end{align}
and finally for the third order field we get
\begin{align}\label{eq:expect_mu}
\E \frz_{n,3} =& 8 \caM_1 \int dz_1dz_2 Y_n(z_2)Y_n(z_1-z_2)\frC_n(z_1-z_2)\theta(t_1-t_2)\frC_n(z_1)\\ & + 2 \caM_2\int dz_1dz_2 Y_n(z_1)Y_n(z_2)\frC_n(z_1-z_2)^2- m_2 \log L^{N} - m_3
\non
\end{align}
where $ \theta(t)= \funit_{t \geq 0}$ is the Heaviside fuction and
\begin{align}\label{eq:m1m2}
\frm_{\alpha\beta}&=\sum_{\gamma, \delta}  M^{(\alpha)}_{\gamma \delta} M^{(\delta)}_{\gamma \beta} \\
(\caM_1)_{\alpha} & = \sum_{\beta_1 \beta_2 \beta_3 \beta_4}M^{(\alpha)}_{\beta_1 \beta_2}M^{(\beta_2)}_{\beta_3 \beta_4} M^{(\beta_4)}_{\beta_1 \beta_3}  \\\nonumber
(\caM_2)_{\alpha} & = \sum_{\beta_1 \beta_2 \beta_3 \beta_4}M^{(\alpha)}_{\beta_1 \beta_2}M^{(\beta_2)}_{\beta_3 \beta_4} M^{(\beta_1)}_{\beta_3 \beta_4} . 
\end{align}
Define  the random field $$
\omega_{\alpha\beta}:=\vartheta_\alpha\vartheta_\beta-\E\vartheta_\alpha\vartheta_\beta
$$ 
(here and below $\vartheta=\vartheta_n^{(N)}$). 
Then the local fields $\zeta_n$ 
are linear combinations of their expectations and the following random fields 
\begin{align}\label{eq:fields1}
&\vartheta_\alpha, \ \ \omega_{\alpha\beta},   \ \ Y_n \ast \omega_{\alpha\beta},  \ \  Y_n \ast \omega_{\alpha\beta}
Y_n \ast\omega_{\gamma\delta} - \E Y_n \ast \omega_{\alpha\beta}
Y_n \ast\omega_{\gamma\delta},\\&\label{eq:fields2}
\ \ \vartheta_\alpha Y_n \ast \omega_{\beta\gamma}
, \ \ \vartheta_\alpha Y_n \ast  (\vartheta_\beta Y_n \ast \omega_{\gamma\delta})-\E \vartheta_\alpha Y_n \ast  (\vartheta_\beta Y_n \ast \omega_{\gamma\delta})
\end{align}
where we used $Y_n \ast\delta_n=0$, while for the bi-local fields we need to consider
\begin{align}
Y_n(z-z')\vartheta_\alpha(z)\vartheta_\beta(z')-\E Y_n(z-z')\vartheta_\alpha(z)\vartheta_\beta(z') \label{bilocal}
\end{align} 
To get the covariance estimates for the fields \eqref{eq:fields1}, \eqref{eq:fields2}, \eqref{bilocal} claimed in Proposition \ref{pro: proba1} 
we need to introduce the mixed covariance $ \frC'_n(z)$ such that 
\beq
\delta_{\alpha \beta}\frC'_n(z):= \bbE \vartheta'_{\alpha}(z) \vartheta_{\beta}(0)
\label{eq:CNN'def }
\eeq
where, as before, the primed kernels and fields have the lower cutoff $\chi'$. Furthermore, let us define 
\baq
 \caC_n(z)&:=& \sup_{N\geq n} |\frC'_n(z)|
\non
\\
 \delta \caC_n(z)&:=& |\frC'_n(z)-\frC_n(z)|
\non
\\
\caY_n(z)&:=&\sup_{N\geq n} |Y_n(z)| 
\non
\\
\delta \caY_n(z)&:=&|Y'_n(z) -Y_n(z) |
\non
\eaq
The regularity of these kernels is summarized in the following Lemma proven in the Appendix.

\begin{lemma}\label{lem: greg} 
\begin{enumerate}
\item[(a)] For $p<3$ and  uniformly in $n$ one has $ \caC_n\in L^p(\bbR\times\bbT_n)$ and
\begin{align}
\|\delta \caC_n\|_p^p\leq C L^{-\la_p(N-n)}\|\chi-\chi'\|_\infty \label{eq: cacnnbound}
\end{align}
for some $ \la_p >0$.
\item[(b)] For $p<\frac{3}{2}$ and uniformly in $n$ one has $  \caY_n\in L^p(\bbR\times\bbT_n)$ and
\begin{align}
\|\delta \caY_n\|_p^p\leq C L^{-\la_p(N-n)}\|\chi-\chi'\|_\infty.\label{eq: cagnnbound}
\end{align}
for some $ \la_p >0$.
\end{enumerate}
 
\end{lemma} 

Having these technical tools at hand, we can finally start to show the covariance estimates.

\subsection{Fields \eqref{eq:fields1} and \eqref{bilocal}}

For $ z=(t,x)$ we will use the norm $|z|=|t|+|x| $ and we will drop the subscript $ n$ from the random fields and kernels.  
From the definition of the smoothing kernel $ \tilde K$ we note that
\begin{align}\label{eq:decayK}
\tilde K(z,z') \leq C e^{-\hf|z-z'|}, \ \ \ \partial_x \tilde K(z,z') \leq C e^{-\hf|z-z'|}, \ \ \ \partial_{x'} \tilde K(z,z') \leq C e^{-\hf|z-z'|}.
\end{align} 
 
Defining $$ X(z_1-z_2):= \E \zeta(z_1)\zeta(z_2)$$
we then get 
\begin{align}\label{eq:zeta_bound}
\E \norm{\rho(z)}^2 & = \int dz_1 dz_2 \tilde K(z,z_1) \tilde K(z,z_2)   X(z_1-z_2)  \leq C \norm{ X}_1
\end{align}
i.e. it suffices to bound the $L^1$-norm of the covariance. We will use repeatedly the Young  inequality in the form
\begin{align}\label{eq:young}
\norm{f _1\ast f_2\ast\dots\ast f_m}_p \leq \prod_{i=1}^m \norm{f_i}_{p_i}
\end{align}
if $n-1+\frac{1}{p}=\sum\frac{1}{p_i}$ where  $ 1 \leq p,p_i \leq \infty$. 
We consider now the fields one by one. 

\vskip 2mm

\noindent (i) For $\zeta=\vartheta_\alpha$ we have $ \norm{X}_1 \leq C \norm{ \caC}_1  $. 

\vskip 2mm

\noindent (ii) For $\zeta=\omega_{\alpha\beta}$ we have $ \norm{X}_1 \leq C \norm{ \caC}_2^2  $.

\vskip 2mm

\noindent (iii) For $\zeta= Y\ast\omega_{\alpha\beta}$ let $Y^t(z)=Y(-z)$. Then
$
X=C Y\ast\frC\ast Y^t
$. By Young inequality
$$
 \norm{X}_1
\leq C\norm{\caY * \caY * \caC^2}_1\leq C \norm{\caY}_1^2 \norm{\caC}_2^2
$$

\vskip 2mm

\noindent (iv) For $\zeta=
Y\ast \omega_{\alpha\beta}
Y\ast\omega_{\gamma\delta}-\E Y\ast \omega_{\alpha\beta}
Y\ast\omega_{\gamma\delta} $ we get
\begin{align}
& \int dz \E \zeta(z) \zeta(0) \leq C\int dz dz_1\dots dz_4 \caY(z-z_1)\caY(z-z_2)\caY(-z_3)\caY(-z_4)\\\non 
& \times [\caC(z_1-z_3)\caC(z_2-z_4)(\caC(z_1-z_4)\caC(z_2-z_3)+\caC(z_1-z_2)\caC(z_3-z_4)) \\\non
&+\caC^2(z_1-z_3)\caC^2(z_2-z_4) + \caC^2(z_1-z_4)\caC^2(z_2-z_3)]
\non
\end{align}
Using the trivial inequality 
\begin{align}\label{eq:simple_dis}
 2|ab| \leq a^2 + b^2
\end{align}
with $a, b \in \R$ for the products of $ \caC$, we obtain
\begin{align}\label{eq:tricky_graphs}\non
\int dz \E \zeta(z) \zeta(0) \leq C & \int dz dz_1\dots dz_4 \caY(z-z_1)\caY(z-z_2)\caY(-z_3)\caY(-z_4)  \\\non 
& \times [\caC(z_1-z_3)\caC(z_2-z_4) (\caC(z_1-z_4)^2+\caC(z_2-z_3)^2) \\\non
& +\caC(z_1-z_2)\caC(z_3-z_4) (\caC(z_1-z_3)^2+\caC(z_2-z_4)^2) \\
&+\caC^2(z_1-z_3)\caC^2(z_2-z_4) + \caC^2(z_1-z_4)\caC^2(z_2-z_3)]
\end{align}
Note that $ \caC_n \in L^p$ with $ p<\frac{3}{2}$ thanks to Lemma \ref{lem: greg}, so by Young inequality one can see that the first two terms in \eqref{eq:tricky_graphs} are bounded by


\begin{align}\non
&  C \| (\caY \ast \caC) (\caY \ast (\caC^2_n (\caY \ast \caY \ast \caC)) \|_1 \leq C \| \caY \ast \caC \|_2  \| \caY \ast \caC^2 (\caY \ast \caY \ast \caC)     \|_2 \\\non
& \leq C  \| \caY \|_1 \| \caC  \|_2 \| \caY \|_{\frac{4}{3}} \| \caC^2 \|_{\frac{4}{3}} \| \caY \ast \caY \ast \caC \|_{\infty} \leq C \| \caY \|_1 \| \caY \|^3_{\frac{4}{3}} \| \caC  \|_2^2  \| \caC^2 \|_{\frac{4}{3}} 
\end{align}
while the third and fourth term in \eqref{eq:tricky_graphs} are bounded by
\begin{align}
&  C \| \caC(\caY \ast \caY) \|_1 \| \caC^2 \ast (\caY (\caY \ast \caC)) \|_1 \leq C \| \caC \|_2 \| \caY \ast \caY \|_2 \| \caC^2 \|_1 \| \caY (\caY \ast \caC) \|_1  \\\non
& \leq C \| \caC \|^4_2 \| \caY \|^4_{\frac{4}{3}} 
\end{align}
and the last constributions are bounded by
\begin{align}
&  C \| \caY \ast \caY \ast \caY \ast \caC^2 \|_2 \| \caC^2 \ast \caY \|_2 \leq C \| \caY \|_1^2  \| \caY \|_{\frac{4}{3}}^2 \| \caC^2 \|^2_{\frac{4}{3}} . \non
\end{align}
\vskip 2mm
\noindent (v) Next we consider the bi-local field $\zeta(z_1,z_2)=
Y(z_1-z_2)(\vartheta_\alpha(z_1)\vartheta_\beta(z_2)-\delta_{\alpha \beta}\frC(z_1-z_2))$. Then we have 
\begin{align}\non
\E\zeta(z_1,z_2)\zeta(z_3,z_4) \leq C\caY(z_1-z_2)\caY(z_3-z_4)(\caC(z_1-z_3)\caC(z_2-z_4)+\caC(z_1-z_4)\caC(z_2-z_3))
\end{align}
so that 
\begin{align}\nonumber
& e^{c|z_1-z_2|} \E \norm{\rho_n(z_1,z_2)}^2  \leq C  e^{c|z_1-z_2|} \bigg| \int dz'_{1234} \tilde K(z_1,z'_1)\tilde K(z_2,z'_2) \tilde K(z_1,z'_3)\tilde K(z_2,z'_4) \\\nonumber
&  \times  \caY(z'_1-z'_2)\caY(z'_3-z'_4)[\caC(z'_1-z'_3)\caC_n(z'_2-z'_4) + \caC_n(z'_1-z'_4)\caC_n(z'_2-z'_3)] \bigg|
\end{align}
where $ 0 < c < \hf$ and then 
\begin{align}\nonumber
& \E \norm{\rho_n(z_1,z_2)}^2 \leq C e^{-c|z_1-z_2|} \| \tilde{\caY} * \caY * \caC \ast \caC  \|_1 \leq C e^{-c|z_1-z_2|} \| \tilde \caY \|_1 \| \caY  \|_1 \| \caC  \|_1^2 
\end{align}
where $\tilde \caY(z):=e^{c|z|} \caY(z)$ is in $  L^p$ with $ p < \frac{3}{2}$.

\subsection{Fields \eqref{eq:fields2} and \eqref{bilocal2} }

We observe that in the above covariance estimates, the Young inequality trick requires all the kernels to be at least in $ L^1(\R \times \bbT_n)$. Unfortunately in the fields \eqref{eq:fields2} and \eqref{bilocal2} the kernel $J_n(z):=Y_n(z)\frC_n(z) $ will appear and it is easy to see that $ \| J_n \|_1$ diverges logarithmically as $ N \to 0$, so Young inequality cannot be applied as before.

The following Lemma shows some properties of $ J_n$ which are crucial to overcome this problem. Its proof can be found in the Appendix.

\begin{lemma}\label{lem:deco_C}  
(a) We have 
\begin{align}
J_n(z) = \partial_x W_n(z)+j_n(z)
\end{align}
where $ W_n$ is in $L^1(\R \times \bbT_n)$ uniformly in $n,N$ and 
$$
|j_n(z)|\leq Ce^{-|x|}\funit_{[0, 2]}(t).
$$
\vskip 2mm
\noindent (b) The function $Z_n:=Y_n\ast J_n$ is in $L^1(\R \times \bbT_n)$ uniformly in $n,N$.\vskip 2mm
\noindent (c) $\| W_n- W'_n\|_1\leq CL^{-\la(N-n)}$ for some $ \la >0$, idem for $j_n$ and $Z_n$.
\vskip 2mm
\noindent (d) Let be $ \epsilon=L^{-2(N-n)}$, then  
$$
|W_n(z)- \frC_n(z)^2|\leq C(\epsilon^{-2}e^{-c|x|/\epsilon}\funit_{[\hf\epsilon^2,2\epsilon^2]}(t)+e^{-c|x|}\funit_{[\hf,2]}(t))
$$
$$
|Y_n(z)-2\partial_x \frC_n(z)|\leq C(\epsilon^{-2}e^{-c|x|/\epsilon}\funit_{[\hf\epsilon^2,2\epsilon^2]}(t)+e^{-c|x|}\funit_{[\hf,2]}(t)).
$$
\end{lemma} 

In practice Lemma \ref{lem:deco_C} guarantees that the nasty kernel $ J_n$ is actually a gradient of an $ L^1$-function, up to to a smooth correction. By an integration by parts this property will allow us to move the gradient and make it act on the smoothing kernels $ \tilde K$, so that we can still use the Young inequality. Moreover, we point out that the kernel $ Z_n$ will appear in the last fields in \eqref{eq:fields2}. Finally, item (d) in Lemma \ref{lem:deco_C} will be employed to study the divergence of $ \E \frz_{n,3}$ in \eqref{eq:expect_mu}.

In the following we will neglect the remainder term $ j(z)$, since its contributions can be easily bounded as we have done for the fields in \eqref{eq:fields1}.

\vspace*{2mm}

\noindent (vi) For $ \zeta = \vartheta_{\alpha} Y \ast \omega_{\beta \gamma}$ we have
\begin{align}
\E \| \rho(z) \|^2 \leq & C \bigg| \int dz_{1234} \tilde K(z,z_1) \tilde K(z,z_2) J(z_1-z_3)J(z_2-z_4) \frC(z_3-z_4)   \bigg| \\\non
& + C \int dz dz_1 dz_2 \caY(z-z_1)\caY(-z_2)[\caC(z)\caC^2(z_1-z_2) + \caC(z-z_2)\caC(z_1)\caC(z_1-z_2)]
\end{align}
Using Lemma \ref{lem:deco_C}, \eqref{eq:decayK} and Young inequality we get
\begin{align}
\E \| \rho(z) \|^2 \leq & C[ \| W \|_1 \| \caC \ast W \|_1 + \| \caC \ast \caY \|_2 \| \caC^2 \ast \caY \|_2 + \| \caC(\caC \ast \caY) \ast \caY \|_2 \| \caC \|_2] \\\non
 \leq & C [ \| W \|_1^2 \| \caC \|_1  + \| \caY \|_1 \| \caC \|_2 \| \caC^2 \|_{\frac{4}{3}} \| \caY \|_{\frac{4}{3}} + \| \caY \|_1 \| \caY \|_{\frac{4}{3}}  \| \caC \|_2 \| \caC^2 \|_{\frac{4}{3}}].
\end{align}

 For the last field in \eqref{eq:fields2}, i.e. $ \vartheta_\alpha Y \ast  (\vartheta_\beta Y \ast \omega_{\gamma\delta})-\E \vartheta_\alpha Y \ast  (\vartheta_\beta Y \ast \omega_{\gamma\delta})$ it is convenient to perform an expansion in terms of Wick polynomials to keep track of the several contributions involved in the covariance (see \cite{LM16} for a recent review about Wick polynomials). In our case the ``elementary'' fields $ \vartheta_n$ are Gaussian and with vanishing expectation value, so the combinatorics of the Wick expansion will be quite simple. Noting that $ \omega_{\alpha \beta}= \w{\vartheta_{\alpha} \vartheta_{\beta}}$, the random fields turns out to be a linear combination of the following terms
\begin{align}\label{eq:wick_exp}\non
& \w{\vartheta_\alpha Y \ast  (\vartheta_\beta Y \ast \vartheta_\gamma \vartheta_{\delta})}, \ \ \ Z \ast  \omega_{\alpha \beta}, \ \ \ \w{\vartheta_{\alpha} Z \ast \vartheta_{\beta}}, \\
&  \int dz_1dz_2 Y(z-z_1)Y(z_1-z_2) \frC(z-z_2) \w{\vartheta_{\alpha}(z_1)\vartheta_{\beta}(z_2)}.
\end{align}

\vspace*{2mm} 

\noindent (vii) In $ \zeta = \w{\vartheta_\alpha Y \ast  (\vartheta_\beta Y \ast \vartheta_\gamma \vartheta_{\delta})}$ there is no $ J$ appearing, so we can just estimate the corresponding $ \| X \|_1$ which unfortunately has many terms:

\begin{align}\label{eq:mu1}\non
& \| X \|_1 \leq C \int dz dz_{1} \cdots dz_4 {\caY}(z-z_1) {\caY}(z_1-z_2) {\caY}(-z_3) {\caY}(z_3-z_4) \\\nonumber
& \times [\mathcal C(z)\mathcal C(z_1-z_3)  \mathcal C^2(z_2-z_4)   + \mathcal C(z-z_3)\mathcal C(z_1)  \mathcal C^2(z_2 - z_4) \\\nonumber
& + \mathcal C(z)\mathcal C(z_1-z_4)\mathcal C(z_2-z_3)\mathcal C(z_2-z_4) + \mathcal C(z-z_3)\mathcal C(z_1-z_4)\mathcal C(z_2)\mathcal C(z_2-z_4) \\\nonumber 
& + \mathcal C(z-z_4)\mathcal C(z_1)\mathcal C(z_2-z_3)\mathcal C(z_2-z_4)  + \mathcal C(z-z_4)\mathcal C(z_1-z_3)\mathcal C(z_2)\mathcal C(z_2-z_4) \\\non
&+ \mathcal C(z-z_4)\mathcal C(z_1-z_4)\mathcal C(z_2-z_3)\mathcal C(z_2)].
\end{align}
Using \eqref{eq:simple_dis} we get
\begin{align}\non
& \| X \|_1 \leq C [\| \caC(\caY \ast \caY \ast \caC^2) \|_{\frac{4}{3}} \| \caC \ast \caY \|_4 + \| \caC \|_2 \| \caY \ast (\caY \ast \caC)(\caY \ast \caY \ast \caC^2 )\|_2 \\\non
& + \| \caC \ast \caY \|_4 \| \caY \ast (\caY \ast \caC)(\caY \ast \caC^2)  \|_{\frac{4}{3}} + \| \caY \ast \caC \|_2 \| \caY \ast \caY \ast \caC (\caY \ast \caC^2) \|_2 \\\non
& + \| \caY \ast \caC^2 \|_2 \| \caY \ast (\caC \ast \caY)^2 \|_2 + \| \caY \|_1 \| (\caC \ast \caY)(\caY \ast \caC(\caC^2 \ast \caY) ) \|_1 \\\non
& + \| \caY \ast \caC^2 \|_{2} \| \caY \ast \caC(\caY \ast \caY \ast \caC) \|_2 + \| \caY \|_1 \| \caC(\caY \ast (\caC \ast \caY)(\caC^2 \ast \caY))  \|_1 \\\non
& + \| \caY  \|_1 \| (\caY \ast \caC^2)(\caY \ast \caC(\caC \ast \caY)) \|_1 + \| \caC^2 \caY \|_2 \|  \caY \ast \caY \ast \caC(\caC \ast \caY)\|_2\\\non
& \leq C [ \| \caY\|_{\frac{4}{3}}^4 \| \caC \|_2^4 + \| \caY \|_1 \| \caY\|_{\frac{4}{3}}^3 \| \caC \|_2^2 \| \caC^2 \|_{\frac{4}{3}}].
\end{align}

\vspace*{2mm}
\noindent (viii) For $ \zeta = Z \ast  \omega_{\alpha \beta}$ we have by Lemma \ref{lem:deco_C}(b) 
\begin{align}
\| X  \| \leq C  \| Z \ast \caC^2 \|_1\| Z \|_1 \leq C \| Z \|_1^2 \| \caC^2 \|_1.
\end{align}

\vspace*{2mm}
\noindent (ix) For $ \zeta = \w{\vartheta_{\alpha} Z \ast \vartheta_{\beta}}$ again by Lemma \ref{lem:deco_C}(b) we have
\begin{align}
\| X  \| \leq C \| Z \ast \caC \|_2^2 \leq C \| Z \|_1^2 \| \caC \|_2^2 . 
\end{align}

\vspace*{2mm}
\noindent (x) For $ \zeta =  \int dz_1dz_2 Y(z-z_1)Y(z_1-z_2) \frC(z-z_2) \w{\vartheta_{\alpha}(z_1)\vartheta_{\beta}(z_2)}$ by \eqref{eq:simple_dis} we have
\begin{align*}
\| X \|_1 \leq C & \int dz dz_1 \cdots dz_4 \caY(z-z_1)\caY(z_1-z_2)\caC(z-z_2)\caY(-z_3)\caY(z_3-z_4)\caC(-z_4) \\\non
& \times [\caC^2(z_1-z_3) + \caC^2(z_2-z_4) +\caC^2(z_1-z_4) +\caC^2(z_2-z_3) ] \\\non
\leq C & \| \caY \|_{\frac{4}{3}}^4 \| \caC \|_2^4.
\end{align*}

\vspace*{2mm} 
\noindent (xi) We still need to bound $ \zeta(z_1-z_2) = \E \sigma_{\alpha \beta}(z_1,z_2)= \frm_{\alpha \beta} J(z_1-z_2)$. Using Lemma \ref{lem:deco_C}(a) and the same strategy used for (v) we get
\begin{align}
e^{c|z_1-z_2|} \| \rho(z_1,z_2) \|^2 \leq C \| \tilde W \|_1^2
\end{align}
where $ \tilde W(z)= e^{c|z|}W(z) \in L^1(\R \times \bbT_n)$.

We observe that the estimates \eqref{eq: twopoint1} and \eqref{eq: twopoint3} are obtained as the bounds \eqref{eq: twopoint} and \eqref{eq: twopoint2} derived above by using Lemma \ref{lemma: heatkernel}, \eqref{eq: cacnnbound}, \eqref{eq: cagnnbound} and Lemma \ref{lem:deco_C}(c).

\subsection{Third order renormalization}

So we are left with the analysis of the expectation  $ \E \frz_{3}$ which will allow us to determine the renormalization constants $ m_2,m_3$:
\begin{align*}
\E \frz_{3} =&2\caM_2\int (Y\ast \frC^2)(z) Y(z)dz
 +8\caM_1\int (Y\ast \theta J)(z)\frC(z)dz-m_2 \log L^N -m_3\\
 =& 4\caM_2\int (Y\ast\theta \frC^2)(z)Y(z)dz+
 8\caM_1\int (Y\ast \theta J)(z)\frC(z)dz -m_2 \log L^N -m_3
\end{align*}
where $(\theta\frC^2)(z)=\theta(t)\frC^2(z)$ and similarly for $ \theta J$. Let us call $ \caA = \partial_x \frC - \hf Y$ and $ \caB=W-\frC^2$. Using Lemma \ref{lem:deco_C} and an
 integration by parts in the second term we get 
\begin{align}
\E \frz_{n,3}=& 8\caM_1\int[ (Y\ast \theta j)(z)\frC(z)-(Y\ast \theta \frC^2)(z)\caA(z)-(Y \ast \theta \caB)(z)\frC(z)]dz -m_3 \nonumber
 \\
& + 4(\caM_2-\caM_1)\int (Y\ast\theta \frC^2)(z) Y(z)dz -m_2 \log L^N 
\label{dive}
\end{align}
For the first term we use the bounds in  Lemma \ref{lem:deco_C} to get 
\begin{align}
\bigg| \int[ (Y\ast \theta j)(z)\frC(z)-(Y\ast \theta \frC^2)(z)\caA(z)-(Y \ast \theta \caB)(z)\frC(z)]dz \bigg| \leq C.
\end{align}

Let us now study the second term in \eqref{dive} which is the divergent one. In Fourier space we have
\begin{align}
& \int (Y\ast\theta \frC^2)(z) Y(z)dz  \\\non
 = & \int_0^{\infty} dt \int dp \, \widehat Y(t,-p) \int_0^t ds \, \widehat Y(t-s,p) \int dq \, \widehat \frC(s,p+q) \widehat \frC (s,q) \\\non
 = & \int_0^{\infty} dt \,\chi_{\epsilon}(t) \int dp \, p^2 e^{- t p^2} \int_{0}^t ds \, \chi_{\epsilon}(t-s) e^{-(t-s)p^2} \int dq \, e^{-s(q^2 + (p+q)^2)} \\\non
& \times h_{\epsilon}(s, \sqrt{s}(p+q)) h_{\epsilon}(s, \sqrt{s}q)
\end{align} 
where
\begin{align}\label{eq:h_eps}
h_\epsilon(t,p)=p^2 \int_0^{\infty} d\sigma \, e^{-2\sigma p^2} \chi_{\epsilon}((1+\sigma)t)\chi_{\epsilon}(\sigma t)
\end{align}
and $ 0 \leq h_{\epsilon} < \hf$ uniformly on $\R_+ \times \R$.
Let us define $ \mu_{\epsilon}$ as
\begin{align}
\mu_{\epsilon} := \frac{1}{4} \int_{\epsilon^2}^{2} dt \, \int dp \, p^2 e^{- t p^2} \int_0^t ds \, e^{-(t-s)p^2}\funit_{[\epsilon^2,2]}(t-s) \int dq \, e^{-s(q^2 + (p+q)^2)} .
\end{align}
We get that 
\begin{align}\non
 0  \leq &  \,\mu_{\epsilon} - \int (Y\ast\theta \frC^2)(z) Y(z)dz   \\\non
\leq & \,\frac{1}{4} \int_{\epsilon^2}^{2} dt \, \int dp \, p^2 e^{- t p^2} \int_0^t ds \, e^{-(t-s)p^2} \int dq \, e^{-s(q^2 + (p+q)^2)} \\
& \times [\funit_{[\epsilon^2,2]}(t)(\funit_{[\epsilon^2,2\epsilon^2]}(t-s)+\funit_{[1,2]}(t-s)) + \funit_{[\epsilon^2,2]}(t-s)(\funit_{[\epsilon^2,2\epsilon^2]}(t)+\funit_{[1,2]}(t))] \leq C
\label{eq:diff_mu_graph}
\end{align}
Let us also define $ \tilde \mu_{\epsilon}:= \frac{\pi}{4\sqrt{3}}\log\epsilon^{-1}$, then by an explicit computation one gets 
\begin{align}\label{eq:diff_mu}
\lim_{\epsilon \to 0} (\tilde \mu_{\epsilon}-\mu_{\epsilon})= \caO(1).
\end{align}
 Therefore, we can identify the universal renormalization constant $ m_2$ as
\begin{align}
m_2 = 4(\caM_2-\caM_1)\frac{\tilde \mu_{\epsilon}}{\log \epsilon^{-1}} = \frac{\pi}{\sqrt{3}}(\caM_2-\caM_1).
\end{align}
Finally, for the $ \chi$-dependent renormalization constant $ m_3$, let be $ \nu_{\epsilon}:= \E \frz_{n,3}-m_2 \log \epsilon^{-1}$: from \eqref{eq:mu1}, \eqref{eq:diff_mu_graph} and \eqref{eq:diff_mu} we know that $|\nu_\epsilon| \leq C$ and by bounds similar to ones in Lemma \ref{lem:deco_C} comparing different cutoffs one can see that $ \nu_{\epsilon}$ is a Cauchy sequence. Therefore, in the end we obtain 
\begin{align}
m_3 = \lim_{\epsilon \to 0} (\E \frz_{n,3}-m_2 \log \epsilon^{-1}) = \caO(1).
\end{align}

\begin{remark}[Cancellation of the third order divergence]
We observe that for some special class of vectors of symmetric matrices $ M=(M^{(1)},M^{(2)},M^{(3)})$ the normalization constants $ m_2$ and $ m_3$ are not needed, i.e. $ m_2=m_3=0$ (for example this is the case of the ordinary KPZ equation where $ u \in \R$). 

In fact, if $ M^{(\alpha)}_{\beta \gamma}$ is totally symmetric with respect to three indices, i.e. it is also invariant under the swap $ \alpha \leftrightarrow \beta$, then $ \caM_1=\caM_2$ in \eqref{eq:expect_mu} and \eqref{eq:m1m2} and the divergent term is not present.

\end{remark}

\appendix
\section{Proof of Lemma \ref{lemma: heatkernel}}

From \eqref{eq: fracCn-ndef} one has
$$
\bbE (\vartheta^{(N)}_n(t,x),M^{(\alpha)}\vartheta^{(N)}_n(t,x)) = \bigg(\sum_{\beta=1}^3 M^{(\alpha)}_{\beta \beta} \bigg) \frC^{(N)}_n(0,0)
$$
Let us split $\mathfrak{C}^{(N)}_n(0,0) $ by isolating the term corresponding to $ i=0$ in \eqref{eq:heat_torus}:
\begin{align}\label{eq:Cn_split}
\mathfrak{C}^{(N)}_n(0,0) & = \frac{1}{2^{7/2} \sqrt{\pi}}\int_0^{\infty} \frac{\chi(s)^2 - \chi'(L^{2(N-n)} s)^2}{s^{3/2}}\,ds + R.
\end{align}
where to stress the cutoff dependence we wrote this with the lower cutoff $\chi'$.

The remainder is easily bounded by
$$
R
\leq C e^{-cL^{2n}}.
$$
and its change with cutoff by
\begin{align*}
|R-R' |  
& 
 \leq C e^{-cL^{2N}}\|\chi - \chi'\|_{\infty}.
\end{align*} 
For the main term in \eqref{eq:Cn_split} we define 
$$
\rho_\chi=\int_0^{\infty} \frac{1 - \chi( s)^2}{s^{3/2}}\,ds .
$$
Then
\begin{align}\nonumber
&  \int_0^{\infty} \frac{\chi(s)^2 - \chi'(L^{2(N-n)} s)^2}{ s^{3/2}}\,ds  =L^{N-n}\rho_{\chi'}-\rho_\chi
\end{align}
Setting $\delta_n^{(N)}=\sum_{\beta=1}^3 M^{(\alpha)}_{\beta \beta}(R-\rho_\chi)$ the claim follows.
\qed

\section{Proof of Lemma \ref{lem: greg}}

(a) We have:
\beq
\frC'_n(t, x)=-\Delta \int_0^{\infty}H_n(t+2s,x)\chi_{N-n}(t+s)\chi'_{N-n}(s)ds
\label{eq:CNN'defi }
\eeq
where $\chi'_{N-n}(t)=\chi(t)-\chi(L^{2(N-n)}t)$. Therefore, since $\chi_{N-n}(t+s)\chi'_{N-n}(s) \leq \funit_{[0,2]}(s)\funit_{[0,2]}(t)$, one has
\begin{align}\label{eq:C'}
| \frC'_n(t, x)| \leq C \funit_{[0,2]}(t) \sum_{j \in \Z} \ell(t,x+jL^n)
\end{align}
where
\begin{align}\label{eq:bound_ell}
\ell(t,x+jL^n)& \leq C \int_0^2 ds (t+2s)^{-\frac{3}{2}} e^{-\frac{x^2}{4(t+2s)}}[1 + x^2 (t + 2s)^{-1}] \\\nonumber
& \leq C e^{-cx^2}(x^2+t)^{-\hf}[1 + x^2(x^2+t)^{-1}]\funit_{[0,2]}(t) + e^{-cx^2/t}t^{-\frac{3}{2}}[1+x^2 t^{-1}]\funit_{[2,\infty)}(t)
\end{align}
Combining \eqref{eq:C'} with \eqref{eq:bound_ell} one gets
\beq
\caC_n(\tau,x)\leq C e^{-cx^2}(x^2+t)^{-\hf}[1 + x^2(x^2+t)^{-1}]\funit_{[0,2]}(t) \in L^p(\R \times \T_n)
\label{eq: cbound}
\eeq
for $p<3$. To show \eqref{eq: cagnnbound}, note that
$$\chi_{N-n}( t+s)|\chi_{\epsilon}(s)-\chi'_{N-n}(s)|\leq \funit_{[\epsilon^2,2\epsilon^2]}(s) \funit_{[0,2]}(t)
\|\chi-\chi'\|_\infty$$ where $ \epsilon= L^{-(N-n)}$.
Hence 
\beq
\delta\caC_n(t,x)\leq C
 \sum_{j \in \Z}\ell_{N-n}(t,x + jL^n)\funit_{[0,2]}(t)\|\chi-\chi'\|_\infty 
\eeq
where
\beq
\ell_{M}(t,x):=\int_0^{2L^{-2M}}  (t+2s)^{-\frac{3}{2}} e^{-\frac{x^2}{4(t+2s)}}[1 + x^2 (t + 2s)^{-1}]ds=L^{M}
\ell_{0}(L^{2M}t,L^{M}x).
 \label{eq:cMtaux }
\eeq
Hence using \eqref{eq:bound_ell} we have
\begin{align}
\|\ell_{M}(t,x)\funit_{[0,2]}(t)\|_p^p & =L^{-(3-p)M}\|\ell_{0}(t,x)\funit_{[0,2L^{2M}]}(t)\|_p^p \\\nonumber
& \leq C L^{-(3-p)M}\bigg(1+\int_2^{2L^{2M}} t^{\frac{3}{2}(1-p)} dt \bigg)\leq C L^{-\la M}
\end{align}
with $\la>0$ for $p<3$. 

(b) The claim follows with the same strategy employed in item (a).
\qed

\section{Proof of Lemma \ref{lem:deco_C}  }

First of all, we note that we can replace $ H_n$ (the heat kernel on $ \T_n$) by $ H$ (the heat kernel on $ \R$) in  $ {\frC_n}$ and $ \caY_n $. Indeed, letting $\tilde \caK$ denote the kernels $ {\frC_n}$, $ \caY_n $ and $ J_n $ built out of $H$ we get
\begin{align}\nonumber
|\tilde \caK(z)-\caK(z)|\leq Ce^{-|x|}\funit_{[0,2]}(t).
\end{align}
Therefore, in the following proof we will consider kernels built with $ H$ and drop the tildes. 

Let  $ \epsilon = L^{-(N-n)}$ and $ \chi_{\epsilon}=\chi_{N-n}$. We will indicate the scale dependence of the kernels by $ \epsilon$ instead of $ n$, i.e. $ \frC_n = \frC_{\epsilon}$ and so on. We work in Fourier space in the $x$ variable:
\begin{align}\nonumber
\widehat \frC_{\epsilon}(t,p)& = p^2e^{-tp^2} \int_0^{\infty} ds \, e^{-2sp^2} \chi_{\epsilon}(t + s)\chi_{\epsilon}( s)
= e^{-tp^2}h_\epsilon(t,\sqrt{t} p)
\end{align}
where $ h_{\epsilon}$ is defined in \eqref{eq:h_eps} and it is uniformly bounded on $\R_+\times\R$. For $Y_{\epsilon}$ we have
\begin{align}\nonumber
\widehat Y_{\epsilon}(t,p)& = \iu p \, e^{-tp^2}  \chi_{\epsilon}(t ).
\end{align}
Thus
\begin{align}\label{hatJ}
\widehat J_{\epsilon}(t,p)&=\iu \int_\R dq (p+q)e^{-t((p+q)^2+q^2)}h_\epsilon(t,\sqrt{t} q)  \chi_{\epsilon}(t )= \frac{\iu p}{\sqrt{t}}\widehat\caW_\epsilon(t,\sqrt{t}p)
\end{align}
where
$$
\widehat\caW_\epsilon(t,r)=\int_\R dq(1+q/r)e^{-((r+q)^2+q^2)}h_\epsilon( t,q) \chi_{\epsilon}(t ).
$$
$\widehat\caW_\epsilon$  is an entire function in $r$ with
\beq \label{enire}
|\widehat \caW_\epsilon(t,r)|\leq C e^{-c(\Re{r})^2}
\eeq
if $|\Im{r}|\leq 1$ (we used $h( t,q)=h(t,-q)$). Hence in particular the inverse Fourier transform $\caW_\epsilon(t,x)$ is in $L^1(\R)$ uniformly in $t$. We end up with the claim with
$$
W_\epsilon(z)=\frac{1}{{t}}\caW_\epsilon(t,x/\sqrt{t}).
$$
(b) It suffices to study $A_\epsilon=Y_\epsilon \ast \partial_x W_\epsilon$. We get
\begin{align}\nonumber
\widehat A_{\epsilon}(t,p)&=-p^2\int_0^t e^{-(t-s)p^2} \chi_{\epsilon}(t-s )\frac{1}{\sqrt{s}}\widehat W_\epsilon(s,\sqrt{s}p)ds=\frac{1}{\sqrt{t}}\hat a_\epsilon(t,\sqrt{t}p)
\end{align}
with
$$
\hat a_\epsilon(t,p)=-p^2\int_0^1e^{-(1-\sigma)p^2} \chi_{\epsilon}((1-\sigma)t )\frac{1}{\sqrt{\sigma}}\widehat W_\epsilon(\sigma t,\sqrt{\sigma}p)d\sigma.
$$
$\hat a_\epsilon$ is entire satisfying \eqref{enire} and the claim follows.
\vskip 2mm
\noindent (c) These claims follow from
\beq \label{enire1}
|\widehat \caW_\epsilon(t,r)-\widehat \caW'_\epsilon(t,r)|\leq C e^{-c(\Re{r})^2}\funit_{[\hf\epsilon^2,2\epsilon^2]}(t)
\eeq
\vskip 2mm
\noindent (d) Let $B_{\epsilon}=\partial_x \frC_{\epsilon}^2 = 2\frC_{\epsilon}\partial_x\frC_{\epsilon}$. Then
\begin{align}\nonumber
\widehat B_{\epsilon}(t,p)&=2\iu \int_\R dq (p+q)e^{-t((p+q)^2+q^2)}h_\epsilon(t,\sqrt{t} p)h_\epsilon(t,\sqrt{t} q) 
\end{align}
Comparing with \eqref{hatJ} and noting that
$$
|2h_\epsilon( t,\sqrt{t} p)-\chi_{\epsilon}(t )|\leq C(\funit_{[\hf\epsilon^2,2\epsilon^2]}(t)+\funit_{[\hf,2]}(t))
$$
we get
$$
|J_{\epsilon}(z)-\partial_x \frC_{\epsilon}(z)^2|\leq C(\epsilon^{-3}e^{-c|x|/\epsilon}\funit_{[\hf\epsilon^2,2\epsilon^2]}(t)+e^{-c|x|}\funit_{[\hf,2]}(t)).
$$
In the same way we get
$$
|Y_{\epsilon}(z)-2\partial_x \frC_{\epsilon}(z)|\leq C(\epsilon^{-2}e^{-c|x|/\epsilon}\funit_{[\hf\epsilon^2,2\epsilon^2]}(t)+e^{-c|x|}\funit_{[\hf,2]}(t)).
$$
\qed


\begin{thebibliography}{}


\bibitem{hairer} M. Hairer: A theory of regularity structures. Invent. Math. 198(2), 269--504 (2014)

\bibitem{CC}  R. Catellier and K. Chouk: Paracontrolled distributions and the 3-dimensional stochastic quantization equation. ArXiv: 1310.6869 (2013) 

\bibitem{GoJa} P. Gon\c calves and M. Jara: Nonlinear fluctuations of weakly asymmetric interacting particle systems. Arch. Rational Mech. Anal. 212(2), 597--644 (2014)

\bibitem{GIP} M. Gubinelli, P. Imkeller, and N. Perkowski: Paracontrolled distributions and singular PDEs. Forum Math. Pi. 3 (2015)


\bibitem{AK}
A. Kupiainen: Renormalization group and Stochastic PDEs.
Ann. Henri Poincaré. 17(3), 497--535 (2016)

\bibitem{Sp14} H. Spohn: Nonlinear Fluctuating Hydrodynamics for Anharmonic Chains. J. Stat. Phys. 154(5), 1191--1227 (2014)

\bibitem{hairerkpz} M. Hairer: Solving the KPZ equation. Ann. Math. 178(2),  559--664 (2013)

\bibitem{wilson} K. Wilson: The renormalization group and critical phenomena. Nobel Lecture. Rev. Mod. Phys. (1984)


\bibitem{cha}
S.B. Chae: Holomorphy and calculus in normed spaces. Marcel
Decker, New York (1985)


\bibitem{bk} J. Bricmont, A. Kupiainen and G. Lin: Renormalization Group and Asymptotics of Solutions of
Nonlinear Parabolic Equations. 
Comm. Pure Appl. Math. 47, 893--922 (1994) 

\bibitem{bgk} J. Bricmont, K. Gawedzki and A. Kupiainen: KAM theorem and quantum field theory.
Comm. Math. Phys. 201(3), 699--727 (1999)








\bibitem{boga}	Vladimir I. Bogachev: Gaussian Measures. Mathematical Surveys and Monographs.	American Mathematical Society (1998)

\bibitem{nualart}  D. Nualart: The Malliavin calculus and related topics. Probability and its Applications (New York). Springer-Verlag, Berlin, second ed. (2006)

\bibitem{LM16} J. Lukkarinen and M. Marcozzi: Wick polynomials and time-evolution of cumulants. ArXiv: 1503.05851 (2016)

\end{thebibliography}
\end{document}